\documentclass[11pt]{article}
\usepackage{fullpage}
\usepackage{amsmath,amssymb,nicefrac,graphicx,lastpage,array,subfigure,varioref}
\usepackage[noend]{algorithmic}
\usepackage{algorithm}
\usepackage{microtype} % Don't break words
\usepackage[inline,nomargin,draft]{fixme}
\fxsetface{inline}{\bfseries\color{red}}
\usepackage[utf8]{inputenc} % danske tegn
\usepackage{pstricks,pst-node,pst-plot,pstricks-add}
\usepackage{scalefnt}
\usepackage{cite}

\newtheorem{lemma}{Lemma}
\newtheorem{theorem}{Theorem}

\newtheorem{example}{Example}
\newcommand{\qed}{\hfill\ensuremath{\Box}\medskip\\\noindent}
\newenvironment{proof}{\noindent\emph{Proof. }}{}

\newcommand{\gap}[2]{g\{#1,#2\}}
\newcommand{\epos}{\textrm{\it{endpos}}}
\newcommand{\spos}{\textrm{\it{startpos}}}
\newcommand{\ignore}[1]{}

\definecolor{patternColor}{rgb}{1,0,0}
\definecolor{gapColor}{gray}{0}
\definecolor{textColor}{gray}{0.9}
\definecolor{remTextColor}{gray}{0.7}

\definecolor{DTUred}{rgb}{0.7,0,0.15}

\title{String Matching with Variable Length Gaps\thanks{An extended abstract of this paper appeared in proceedings of the 17th Symposium on String Processing and Information Retrieval.}}
\author{Philip Bille \and Inge Li G{\o}rtz \and Hjalte Wedel Vildh{\o}j \and David Kofoed Wind}

\begin{document}

\maketitle
\begin{abstract}
  \noindent We consider string matching with variable length
  gaps. Given a string $T$ and a pattern $P$ consisting of strings
  separated by variable length gaps (arbitrary strings of length in a
  specified range), the problem is to find all ending positions of
  substrings in $T$ that match $P$. This problem is a basic primitive
  in computational biology applications. Let $m$ and $n$ be the lengths
  of $P$ and $T$, respectively, and let $k$ be the number of strings
  in $P$. We present a new algorithm achieving time $O(n\log k + m 
  +\alpha)$ and space $O(m + A)$, where $A$ is the sum of the lower
  bounds of the lengths of the gaps in $P$ and $\alpha$ is the total
  number of occurrences of the strings in $P$ within $T$. Compared to
  the previous results this bound essentially achieves the best known
  time and space complexities simultaneously. Consequently, our
  algorithm obtains the best known bounds for almost all combinations of
  $m$, $n$, $k$, $A$, and $\alpha$. Our algorithm is surprisingly simple
  and straightforward to implement.
  We also present algorithms for finding and encoding the positions of all
  strings in $P$ for every match of the pattern.
\end{abstract}

\section{Introduction}
%\subsection{The Variable Length Gap Problem}
Given integers $a$ and $b$, $0 \leq a \leq b$, a \emph{variable length gap} 
$\gap{a}{b}$ is an arbitrary string over $\Sigma$ of length between $a$ and $b$,
 both inclusive. A \emph{variable length gap pattern} (abbreviated VLG pattern) 
$P$ is the concatenation of a sequence of strings and variable length gaps,
 that is, $P$ is of the form
 $$P = P_1 \cdot \gap{a_1}{b_1} \cdot P_2 \cdot \gap{a_2}{b_2} \cdots
 \gap{a_{k-1}}{b_{k-1}} \cdot P_k \;.$$ A VLG pattern $P$
 \emph{matches} a substring $S$ of $T$ iff $S = P_1 \cdot G_1 \cdots
 G_{k-1} \cdot P_k$, where $G_i$ is any string of length between $a_i$
 and $b_i$, $i = 1, \ldots, k-1$. Given a string $T$ and a VLG pattern
 $P$, the \emph{variable length gap problem} (VLG problem) is to find
 all ending positions of substrings in $T$ that match $P$.

\begin{example}As an example, consider the problem instance over the
  alphabet $\Sigma=\{A,G,C,T\}$:
  \begin{align*}
    T &= \text{ATCGGCTCCAGACCAGTACCCGTTCCGTGGT} \\
    P &= \text{A} \cdot \gap{6}{7} \cdot \text{CC} \cdot \gap{2}{6}
    \cdot \text{GT} 
\end{align*}\label{ex:runningexample}
The solution to the problem instance is the set of positions
$\{17,28,31\}$. For example the solution contains 17, since the
substring ATCGGCTCCAGACCAGT, ending at position 17 in $T$, matches $P$.
\end{example}

Variable length gaps are frequently used in computational biology
applications~\cite{MM1993,Myers1996,NR2003,FG2008,FG2009}. For
instance, the PROSITE data base~\cite{BB1994,Hofmann1999} supports
searching for proteins specified by VLG patterns.

\subsection{Previous Work}
We briefly review the main worst-case bounds for the VLG problem. As
above, let $P = P_1 \cdot \gap{a_1}{b_1} \cdot P_2 \cdot
\gap{a_2}{b_2} \cdots \gap{a_{k-1}}{b_{k-1}} \cdot P_k$ be a VLG
pattern consisting of $k$ strings, and let $T$ be a string. To state
the bounds, let $m = \sum_{i=1}^{k} |P_i|$ be the sum of the lengths
of the strings in $P$ and let $n$ be the length of $T$.

The simplest approach to solve the VLG problem is to translate $P$
into a regular expression and then use an algorithm for regular
expression matching. Unfortunately, the translation produces a regular
expression significantly longer than $P$, resulting in an inefficient
algorithm. Specifically, suppose that the alphabet $\Sigma$ contains
$\sigma$ characters, that is, $\Sigma = \{c_1, \ldots, c_\sigma\}$.
Using standard regular expression operators (union and concatenation),
we can translate $\gap{a}{b}$ into the expression
$$
\gap{a}{b} = \overbrace{C \cdots C}^a\overbrace{(C|\epsilon) \cdots (C|\epsilon)}^{b-a},
$$
\noindent where $C$ is shorthand for the expression $(c_1\mid c_2 \mid
\ldots c_\sigma)$. Hence, a variable length gap $\gap{a}{b}$,
represented by a constant length expression in $P$, is translated into
a regular expression of length $\Omega(\sigma b)$. Consequently, a
regular expression $R$ corresponding to $P$ has length $\Omega(B\sigma
+ m)$, where $B = \sum_{i=1}^{k-1} b_i$ is the sum of the upper bounds
of the gaps in $P$. Using Thompson's textbook regular expression
matching algorithm~\cite{Thomp1968} this leads to an algorithm for the
VLG problem using $O(n(B\sigma + m))$ time. Even with the fastest
known algorithms for regular expression matching this bound can only
be improved by at most a polylogarithmic
factor~\cite{Myers1992,NR2004,Bille06,BT2009}.

Several algorithms that improve upon the direct translation to a
regular expression matching problem have been
proposed~\cite{MM1993,Myers1996,CIMRTT2002,NR2003,LAIP2003,MPVZ2005,RILMS2006,FG2008,FG2009,BT2010}. Some of these are able to solve more 
general versions of the problem, such as searching for 
patterns that also contain character classes and variable 
length gaps with negative length. Most of the algorithms are based on fast simulations of non-deterministic finite automata. In particular, Navarro 
and Raffinot~\cite{NR2003} gave an algorithm using $O(n(\frac{m + B}{w} +
1))$ time, where $w$ is the number of bits in a memory word.
Fredrikson and Grabowski~\cite{FG2008,FG2009} improved this bound for
the case when all variable length gaps have lower bound $0$ and
identical upper bound $b$. Their fastest algorithm achieves
$O(n(\frac{m\log \log b}{w} + 1))$ time. Very recently, Bille and
Thorup~\cite{BT2010} gave an algorithm using $O(n (k
\frac{\log w}{w} + \log k) + m \log m + A)$ time and $O(m + A)$ space,
where $A = \sum_{i=1}^{k-1} a_i$ is the sum of the lower bounds on the lengths of the gaps.
Note that if we assume that the $nk$ term dominates and ignore the
$w/\log w$ factor, the time bound reduces to $O(nk)$.

An alternative approach, suggested independently by Morgante et
al.~\cite{MPVZ2005} and Rahman et al.~\cite{RILMS2006}, is to design
algorithms that are efficient in terms of the total number of
occurrences of the $k$ strings $P_1, \ldots, P_k$ within $T$. Let
$\alpha$ be this number, e.g., in Example~\ref{ex:runningexample} A,
CC, and GT occur $5$, $5$, and $4$ times in $T$. Hence, $\alpha = 5 +
5 + 4 = 14$. Rahman et al.~\cite{RILMS2006} gave an algorithm using
$O(n\log k + m + \alpha\log(\max_{1\leq i < k} (b_i - a_i)))$
time\footnote{The bound stated in the paper does not include the $\log
  k$ factor, since they assume that the size of the alphabet is
  constant. We make no assumption on the alphabet size and therefore
  include it here.}. Morgante et al.~\cite{MPVZ2005} gave a faster
algorithm using $O(n\log k + m + \alpha)$ time. Each of the $k$
strings in $P$ can occur at most $n$ times and therefore $\alpha \leq
nk$. Hence, in the typical case when the strings occur less
frequently, i.e, $\alpha = o(n(k\frac{\log w}{w} + \log k))$, these
approaches are faster. However, unlike the automata based algorithm
that only use $O(m + A)$ space, both of these algorithm use $\Theta(m
+ \alpha)$ space. Since $\alpha$ typically increases with the length
of $T$, the space usage of these algorithms is likely to quickly become
a bottleneck for processing large biological data bases.

\subsection{Our Results}
We address the basic question of whether is it possible to design an
algorithm that simultaneously is fast in the total number of
occurrences of the $k$ strings and uses little space. We show the
following result.
\begin{theorem}\label{thm:vlgproblem}
  Given a string $T$ and a $VLG$ pattern $P$ with $k$ strings, we can
  solve the variable length gaps matching problem in time $O(n\log k + m + \alpha)$ and space $O(m + A)$. Here, $\alpha$ is
  the number of occurrences of the strings of $P$ in $T$ and $A$ is
  the sum of the lower bounds of the gaps. 
\end{theorem}
Hence, we match the best known time bounds in terms of $\alpha$ and
the space for the fastest automata based approach. Consequently,
whenever $\alpha = o(n(k\frac{\log w}{w} + \log k))$ the time and
space bounds of Theorem~\ref{thm:vlgproblem} are the best
known. Our algorithm uses a standard comparison based version of the Aho-Corasick automaton for multi-string
matching~\cite{AC1975}. 
If the size of the alphabet is constant or we use hashing the $\log k$
factor in the running time disappears. Furthermore, our algorithm is surprisingly simple and 
straightforward to implement.

In some cases, we may also be interested in outputting not only the
ending positions of matches of $P$, but also the positions of the individual
strings in $P$ for each match of $P$ in $T$. Note that there can be
exponentially many, i.e., $\Omega(\prod_{i=1}^{k-1} 1+b_i - a_i)$, of these occurrences ending at the same position in $T$. Morgante et al.~\cite{MPVZ2005} showed
how to encode all of these in a graph of size $\Theta(\alpha)$. We show how our algorithm can be extended to efficiently output such a graph. Furthermore, we show two solutions for outputting the positions encoded in the graph. Both solutions use little space since they avoid the need to store the entire graph. The first solution is a black-box solution that works with any algorithm for constructing the graph. The second is a direct approach obtained using a simple extension of our algorithm.

Recently, Haapasalo et al.~\cite{HS2011} studied practical algorithms for an extension of the VLG problem that allows multiple patterns and gaps with unbounded upper bounds. We note that the result of Thm. 1 is straightforward to generalize to this case.

\subsection{Technical Overview}
The previous work by Morgante et al.~\cite{MPVZ2005} and Rahman et
al.~\cite{RILMS2006} find all of the $\alpha$ occurrences of
the strings $P_1, \ldots, P_k$ of $P$ in $T$ using a standard
multi-string matching algorithm (see
Section~\ref{sec:multistringmatching}). From these, they construct a
graph of size $\Omega(\alpha)$ to represent possible combinations of
string occurrences that can be combined to form occurrences of $P$.

Our algorithm similarly finds all of the occurrences of the strings of
$P$ in $T$. However, we show how to avoid constructing a large graph
representing the possible combinations of occurrences. Instead we
present a way to efficiently represent sufficient information to
correctly find the occurrences of $P$, leading to a significant space
improvement from $O(m + \alpha)$ to $O(m + A)$. Surprisingly, the
algorithm needed to achieve this space bound is very simple, and only
requires maintaining a set of sorted lists of disjoint intervals. Even
though the algorithm is simple the space bound achieved by it is
non-obvious. We give a careful analysis leading to the $O(m + A)$
space bound.

Our space-efficient black-box solution for reporting the positions of the
individual strings in $P$ for each match of $P$ in $T$ is obtained by
constructing the graph for overlapping chunks of $T$ of size $2(m+B)$.
Hence the solution is parametrized by the time and space complexity of the 
actual algorithm used to construct the graph.

\section{Algorithm}
In this section we present the algorithm. For completeness, we first briefly
review the classical Aho-Corasick algorithm for multiple string
matching in Section~\ref{sec:multistringmatching}. We then
define the central idea of \emph{relevant occurrences} in
Section~\ref{sec:relevantocc}. We present the full algorithm in
Section~\ref{sec:thealgorithm} and analyze it in Section~\ref{sec:analysis}.

\subsection{Multi-String Matching}\label{sec:multistringmatching}
\newcommand{\path}{\ensuremath\mathrm{path}}
\newcommand{\depth}{\ensuremath\mathrm{depth}}
\newcommand{\occ}{\ensuremath\mathrm{occ}}
\newcommand{\AC}{\ensuremath\mathrm{AC}}

Given a set of pattern strings $\mathcal{P} = \{P_1, \ldots, P_k\}$ of total
length $m$ and a text $T$ of length $n$ the \emph{multi-string
  matching problem} is to report all occurrences of each pattern
string in $T$. Aho and Corasick~\cite{AC1975} generalized the
classical Knuth-Morris-Pratt algorithm~\cite{KMP1977} for single
string matching to multiple strings. The \emph{Aho-Corasick automaton}
(AC-automaton) for $\mathcal{P}$, denoted $\AC(\mathcal{P})$, consists
of the trie of the patterns in $\mathcal{P}$. Hence, any path from the root of
the trie to a state $s$ corresponds to a prefix of a pattern in $\mathcal{P}$.
We denote this prefix by $\path(s)$. For each state $s$ there is also
a special \emph{failure transition} pointing to the unique state $s'$
such that $\path(s')$ is the longest prefix of a pattern in $\mathcal{P}$
matching a proper suffix of $\path(s)$. Note that the depth of $s'$ in
the trie is always strictly smaller for non-root states than the depth
of $s$.

Finally, for each state $s$ we store the subset $\occ(s) \subseteq \mathcal{P}$
of patterns that match a suffix of $\path(s)$. Since the patterns in
$\occ(s)$ share suffixes we can represent $\occ(s)$ compactly by
storing for $s$ the index of the longest string in $\occ(s)$ and a
pointer to the state $s'$ such that $\path(s')$ is the second longest
string if any. In this way we can report $\occ(s)$ in $O(|\occ(s)|)$
time. 
 
The maximum outdegree of any state is bounded by the number of leaves
in the trie which is at most $k$. Hence, using a standard
comparison-based balanced search tree to index the trie transitions
out of each state we can construct $\AC(\mathcal{P})$ in $O(m\log k)$ time and
$O(m)$ space.

To find the occurrences of $\mathcal{P}$ in $T$, we read the characters
of $T$ from left-to-right while traversing $\AC(\mathcal{P})$ to
maintain the longest prefix of the strings in $\mathcal{P}$ matching
$T$. At a state $s$ and character $c$ we proceed as follows. If
$c$ matches the label of a trie transition $t$ from $s$, the next
state is the child endpoint of $t$. Otherwise, we recursively follow
failure transitions from $s$ until we find a state $s'$ with a trie
transition $t'$ labeled $c$. The next state is then the child
endpoint of $t'$. If no such state exists, the next state is the root
of the trie. For each failure transition traversed in the algorithm we
must traverse at least as many trie transitions. Therefore, the total
time to traverse $\AC(\mathcal{P})$ and report occurrences is $O(n
\log k + \alpha)$, where $\alpha$ is the total number of occurrences.

Hence, the Aho-Corasick algorithm solves multi-string matching in
$O((n+m)\log k + \alpha)$ time and $O(m)$ space.

\subsection{Relevant Occurrences}\label{sec:relevantocc}

For a substring $x$ of $T,$ let $\spos(x)$ and $\epos(x)$ denote the start and end position of $x$ in $T$, respectively.
Let $x$ be an occurrence of $P_i$ with $\tau = \epos(x)$ in $T$, and let $R(x)$ denote the range $[\tau+a_i+1;\tau+b_i+1]$ in $T$. 

An occurrence $y$ of $P_i$ in $T$ is a \emph{relevant occurrence} of $P_i$ iff $i=1$ or $\spos(y) \in R(x)$, for some relevant occurrence $x$ of $P_{i-1}$. 
See Fig.~\ref{fig:relevantocc} for an example. Relevant occurrences are similar to the \emph{valid occurrences} defined in~\cite{RILMS2006}. The  difference is that a valid occurrence is an occurrence of $P_{i+1}$ that is in $R(x)$ for \emph{any} occurrence $x$ of $P_i$ in $T$, i.e., $x$ need not be a valid occurrence itself. 

\begin{figure}[tb]
\centering
\newpsstyle{occStyle}{%
  fillstyle=vlines,hatchcolor=black!40,
  hatchwidth=0.1\pslinewidth,
  hatchsep=1\pslinewidth,
  linewidth=0.3pt
}

\psset{unit=0.5cm,yunit=0.8}
\begin{pspicture}[showgrid=false](0,-2)(24,6)
\footnotesize
\rput(0,0){\psframe[style=occStyle](3.5,1)}
{\psset{linewidth=0.3pt,linestyle=dashed,dash=1pt 1pt}
\psline{-}(0.5,1)(0.5,0)
\psline{-}(1,1)(1,0)
\rput(1.75,0.5){$P_i$}
\psline{-}(2.5,1)(2.5,0)
\psline{-}(3,1)(3,0)
}

\psline{|-|}(10.25,0.5)(14.75,0.5)\rput(12.5,0){$R(x)$}

\rput(8,1){\psframe[style=occStyle](5,1)}
{\psset{linewidth=0.3pt,linestyle=dashed,dash=1pt 1pt}
\psline{-}(8.5,2)(8.5,1)
\psline{-}(9,2)(9,1)
\rput(10.5,1.5){$P_{i+1}$}
\psline{-}(12,2)(12,1)
\psline{-}(12.5,2)(12.5,1)
}

\rput(11.5,2.5){\psframe[style=occStyle](5,1)}
{\psset{linewidth=0.3pt,linestyle=dashed,dash=1pt 1pt}
\psline{-}(12,3.5)(12,2.5)
\psline{-}(12.5,3.5)(12.5,2.5)
\rput(14,3){$P_{i+1}$}
\psline{-}(15.5,3.5)(15.5,2.5)
\psline{-}(16,3.5)(16,2.5)
}
\rput(15,4){\psframe[style=occStyle](5,1)}
{\psset{linewidth=0.3pt,linestyle=dashed,dash=1pt 1pt}
\psline{-}(15.5,5)(15.5,4)
\psline{-}(16,5)(16,4)
\rput(17.5,4.5){$P_{i+1}$}
\psline{-}(19,5)(19,4)
\psline{-}(19.5,5)(19.5,4)
}

{
\psset{linewidth=0.5pt,arrowscale=1.5}
\psline{->}(1,3)(1.5,1.4)\rput(0.9,3.4){$x$}
\psline{->}(6.9,1.9)(7.8,1.5)\rput[r](7,2){Not relevant~}
\psline{->}(9.9,3.4)(11.3,3)\rput[r](10,3.5){Relevant~}
\psline{->}(12.9,4.9)(14.7,4.5)\rput[r](13,5){Not relevant~}
\psline{->}(17,1)(15,2.2)\rput(17.3,0.8){$y$}
}
%Axis
\psline[arrowscale=2.5,linewidth=0.8pt]{->}(0,-1)(24,-1)\rput(22.5,-1.8){\scriptsize Position in $T$}

{\psset{linestyle=dashed,linewidth=0.3pt,dash=0.5pt 1pt}
\psline{-}(3.25,0)(3.25,-1)
\psline{-}(19.75,4)(19.75,-1)
}

%Marks
\psline(3.25,-0.8)(3.25,-1.2)\rput(3.25,-1.7){$\tau$}
\psline(10.25,-0.8)(10.25,-1.2)\rput(10.25,-1.7){$\tau+a_i+1$}
\psline(14.75,-0.8)(14.75,-1.2)\rput(14.75,-1.7){$\tau+b_i+1$}
\psline(19.75,-0.8)(19.75,-1.2)\rput(19.75,-1.7){$u$}

\end{pspicture}
\caption{In this figure $x$ is an occurrence of $P_i$ in $T$ reported at position $\tau$. The first and last occurrence of $P_{i+1}$ start outside $R(x)$ thereby violating the $i$th gap constraint, so these occurrences are not relevant compared to $x$. The second occurrence $y$ of $P_{i+1}$ starts in $R(x)$, so if $x$ is itself relevant, then $y$ is also relevant. \label{fig:relevantocc}}
\end{figure}
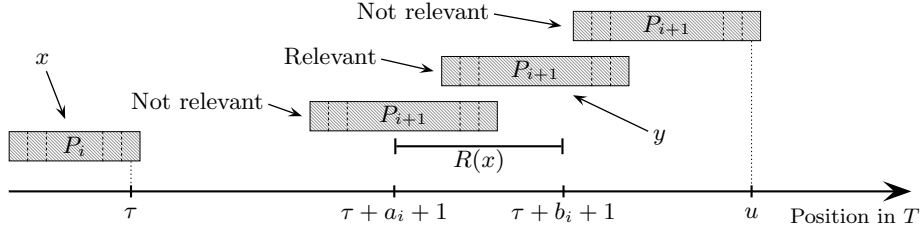

From the definition of relevant occurrences, it follows directly that we can solve the VLG problem by finding the relevant occurrences of $P_k$ in $T$. Specifically, we have the following result.

\begin{lemma} Let $S$ be a substring of $T$ matching the VLG pattern $S_1 \cdot \gap{a_1}{b_1} \cdot S_2 \cdot \gap{a_2}{b_2} \cdots S_k$. Then, $\spos(S_{i+1}) \in R(S_i)$ for all $i=1,\ldots, k-1$.
\end{lemma}

\subsection{The Algorithm}\label{sec:thealgorithm}
Algorithm~\ref{alg:ouralgorithm} computes the relevant occurrences of $P_k$ using the output from the AC automaton. The idea behind the algorithm is to keep track of the ranges defined by the relevant occurrences of each subpattern $P_i$, such that we efficiently can check if an occurrence of $P_i$ is relevant or not. 
More precisely, for each subpattern $P_i$, $i=2,\ldots,k$, we maintain a sorted list $L_i$ containing the ranges defined by previously reported relevant occurrences of $P_{i-1}$. When an occurrence of $P_i$ is reported by the AC automaton, we can determine whether it is relevant by checking if it starts in a range contained in $L_i$ (step~\ref{algstep:relevancecheck}). Initially, the lists $L_2,L_3,\ldots,L_k$ are empty. When a relevant occurrence of $P_i$ is reported, we add the range defined by this new occurrence to the end of $L_{i+1}$. In case the new range $[s,t]$ overlaps or adjoins the last range $[q,r]$ in $L_{i+1}$ ($s \leq r+1$) we merge the two ranges into a single range $[q,t]$. 

Let $\tau$ denote the current position in $T$. A range $[a,b] \in L_i$ is \emph{dead at position} $\tau$ iff $b < \tau - |P_i|$. When a range is dead no future occurrences $y$ of $P_i$ can start in that range since $\epos(y) \geq \tau$  implies $\spos(y) \geq \tau - |P_i|$. 
In Fig.~\ref{fig:relevantocc} the range $R(x)$ defined by $x$ dies, when position $u$ is reached. Our algorithm repeatedly removes any dead ranges to limit the size of the lists $L_2,L_3,\ldots,L_k$. To remove the dead ranges in step~\ref{algstep:removedead} we traverse the list and delete all dead ranges until we meet a range that is not dead. Since the lists are sorted, all remaining ranges in the list are still alive. See Fig.~\ref{fig:runningexample} for an example.

\begin{algorithm}[tb]
\caption{Algorithm solving the VLG problem for a VLG pattern $P$ and a string $T$.\label{alg:ouralgorithm}}

\begin{enumerate}
\item Build the AC-automaton for the subpatterns $P_1,P_2,\ldots,P_k$.
\item Process $T$ using the automaton and each time an occurrence $x$ of $P_i$ is reported at position $\tau=\epos(x)$ in $T$ do:
\begin{enumerate}
        \item\label{algstep:removedead} Remove any dead ranges from the lists $L_i$ and $L_{i+1}$.
        \item\label{algstep:relevancecheck} If $i=1$ or $\tau - |P_i|=\spos(x)$ is contained in the first range in $L_i$ do:
        \begin{enumerate}
                \item\label{algstep:appendrange}If $i<k$: Append the range $R(x)=[\tau + a_i + 1 ; \tau + b_i + 1]$ to the end of $L_{i+1}$. If the range overlaps or adjoins the last range in $L_{i+1}$, the two ranges are merged into a single range.
                \item If $i=k$: Report $\tau$.
        \end{enumerate}
\end{enumerate}
\end{enumerate}
\vspace{-0.3cm}
\end{algorithm}

\ignore{
\subsection{Illustration of the algorithm}
To illustrate how our algorithm works consider the VLG problem in Example~\ref{ex:runningexample}. Figure~\ref{fig:runningexample} shows all the occurrences of the components of the VLG pattern $P$ and the ranges that they define in $T$. At position $\tau=26$, the occurrence $x$ of $P_2$ is reported by the Aho-Corasick automaton. At that point in the execution of the algorithm, the contents of the lists $L_2$ and $L_3$ are as follows:
\[
L_2 ~=~ \bigl[\; [17;20],[22;23],[25;26]\; \bigl] \qquad \mathrm{and} \qquad
L_3 ~=~ \bigl[\; [23,28] \; \bigl]
\]
\noindent In step~\ref{algstep:removedead} of the algorithm, the ranges $[17;20]$ and $[22;23]$ are removed from $L_2$, since they end before $\tau - |P_i|=24$. In step~\ref{algstep:relevancecheck} the algorithm determines that $x$ is relevant because $startpox(x)=25\in [25;26]$. Then $R(x)=[29;33]$ is appended to $L_3$ in step~\ref{algstep:appendrange} and the resulting lists become:
\[
L_2 ~=~ \bigl[\; [25;26]\; \bigl] \qquad \mathrm{and} \qquad
L_3 ~=~ \bigl[\; [23;33] \; \bigl]
\]
\noindent While the algorithm processes the remaining characters of $T$ and discovers the two last relevant occurrences of $P_3$, no further changes are made to $L_2$ and $L_3$.
}

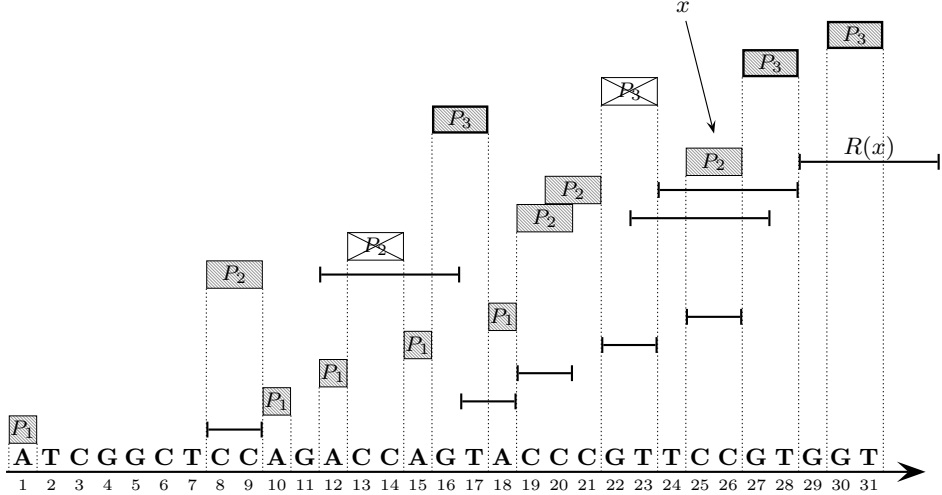
\begin{figure}[tb]
\centering
\newpsstyle{occStyle}{%
  fillstyle=vlines,hatchcolor=black!40,
  hatchwidth=0.1\pslinewidth,
  hatchsep=1\pslinewidth,
  linewidth=0.3pt
}

\psset{unit=0.5cm,yunit=0.75,xunit=1.5,arrowscale=1}
\begin{pspicture}[showgrid=false](1.2,2)(18,19)
\footnotesize
\psline[arrowscale=2.5,linewidth=0.8pt]{->}(1.2,3)(17.5,3)

%Occurrences
{\scriptsize
%P1's
\rput(1.25,4){\psframe[style=occStyle](0.5,1)}
\rput(1.5,4.5){$P_1$}
\psline{|-|}(4.75,4.5)(5.75,4.5)

\rput(5.75,5){\psframe[style=occStyle](0.5,1)}
\rput(6.0,5.5){$P_1$}
\psline{|-|}(9.25,5.5)(10.25,5.5)

\rput(6.75,6){\psframe[style=occStyle](0.5,1)}
\rput(7.0,6.5){$P_1$}
\psline{|-|}(10.25,6.5)(11.25,6.5)

\rput(8.25,7){\psframe[style=occStyle](0.5,1)}
\rput(8.5,7.5){$P_1$}
\psline{|-|}(11.75,7.5)(12.75,7.5)

\rput(9.75,8){\psframe[style=occStyle](0.5,1)}
\rput(10,8.5){$P_1$}
\psline{|-|}(13.25,8.5)(14.25,8.5)

%P2's
\rput(4.75,9.5){\psframe[style=occStyle](1,1)}
\rput(5.25,10){$P_2$}
\psline{|-|}(6.75,10)(9.25,10)

%Not relevant occ:
\rput(7.25,10.5){\psframe[linewidth=0.3\pslinewidth](1,1)}
\psline[linewidth=0.3\pslinewidth]{-}(7.25,10.5)(8.25,11.5)
\psline[linewidth=0.3\pslinewidth]{-}(7.25,11.5)(8.25,10.5)
\rput(7.75,11){$P_2$}

\rput(10.25,11.5){\psframe[style=occStyle](1,1)}
\rput(10.75,12){$P_2$}
\psline{|-|}(12.25,12)(14.75,12)

\rput(10.75,12.5){\psframe[style=occStyle](1,1)}
\rput(11.25,13){$P_2$}
\psline{|-|}(12.75,13)(15.25,13)

\rput(13.25,13.5){\psframe[style=occStyle](1,1)}
\rput(13.75,14){$P_2$}
\psline{|-|}(15.25,14)(17.75,14)

%P3's
\rput(8.75,15){\psframe[style=occStyle,linewidth=3\pslinewidth](1,1)}
\rput(9.25,15.5){$P_3$}

\rput(11.75,16){\psframe[linewidth=0.3\pslinewidth](1,1)}
\psline[linewidth=0.3\pslinewidth]{-}(11.75,16)(12.75,17)
\psline[linewidth=0.3\pslinewidth]{-}(11.75,17)(12.75,16)
\rput(12.25,16.5){$P_3$}

\rput(14.25,17){\psframe[style=occStyle,linewidth=3\pslinewidth](1,1)}
\rput(14.75,17.5){$P_3$}

\rput(15.75,18){\psframe[style=occStyle,linewidth=3\pslinewidth](1,1)}
\rput(16.25,18.5){$P_3$}
}

{\psset{linewidth=0.03,arrowscale=1.5}
\psline{->}(13.25,19)(13.75,15)\rput(13.2,19.5){$x$}
}

\rput(16.5,14.5){$R(x)$}

{
\psset{linestyle=dashed,linewidth=0.3pt,dash=0.5pt 1pt}
\psline{-}(1.25,4)(1.25,3)
\psline{-}(1.75,4)(1.75,3)
\psline{-}(4.75,9.5)(4.75,3)
\psline{-}(5.75,9.5)(5.75,3)
\psline{-}(6.25,5)(6.25,3)
\psline{-}(6.75,6)(6.75,3)
\psline{-}(7.25,10.5)(7.25,3)
\psline{-}(8.25,10.5)(8.25,3)
\psline{-}(8.75,15)(8.75,3)
\psline{-}(9.75,15)(9.75,3)
\psline{-}(10.25,11.5)(10.25,3)
\psline{-}(11.75,16)(11.75,3)
\psline{-}(12.75,16)(12.75,3)
\psline{-}(13.25,13.5)(13.25,3)
\psline{-}(14.25,17)(14.25,3)
\psline{-}(15.25,17)(15.25,3)
\psline{-}(15.75,18)(15.75,3)
\psline{-}(16.75,18)(16.75,3)
}

%Text string
{\bfseries
\newcommand{\searchtextheight}{3.55}
\rput(1.5,\searchtextheight){A}
\rput(2.0,\searchtextheight){T}
\rput(2.5,\searchtextheight){C}
\rput(3.0,\searchtextheight){G}
\rput(3.5,\searchtextheight){G}
\rput(4.0,\searchtextheight){C}
\rput(4.5,\searchtextheight){T}
\rput(5.0,\searchtextheight){C}
\rput(5.5,\searchtextheight){C}
\rput(6.0,\searchtextheight){A}
\rput(6.5,\searchtextheight){G}
\rput(7.0,\searchtextheight){A}
\rput(7.5,\searchtextheight){C}
\rput(8.0,\searchtextheight){C}
\rput(8.5,\searchtextheight){A}
\rput(9.0,\searchtextheight){G}
\rput(9.5,\searchtextheight){T}
\rput(10.0,\searchtextheight){A}
\rput(10.5,\searchtextheight){C}
\rput(11.0,\searchtextheight){C}
\rput(11.5,\searchtextheight){C}
\rput(12.0,\searchtextheight){G}
\rput(12.5,\searchtextheight){T}
\rput(13.0,\searchtextheight){T}
\rput(13.5,\searchtextheight){C}
\rput(14.0,\searchtextheight){C}
\rput(14.5,\searchtextheight){G}
\rput(15.0,\searchtextheight){T}
\rput(15.5,\searchtextheight){G}
\rput(16.0,\searchtextheight){G}
\rput(16.5,\searchtextheight){T}
}

{%Axis labels
\tiny
\newcommand{\axislabelheight}{2.55}
\rput(1.5,\axislabelheight){1}
\rput(2.0,\axislabelheight){2}
\rput(2.5,\axislabelheight){3}
\rput(3.0,\axislabelheight){4}
\rput(3.5,\axislabelheight){5}
\rput(4.0,\axislabelheight){6}
\rput(4.5,\axislabelheight){7}
\rput(5.0,\axislabelheight){8}
\rput(5.5,\axislabelheight){9}
\rput(6.0,\axislabelheight){10}
\rput(6.5,\axislabelheight){11}
\rput(7.0,\axislabelheight){12}
\rput(7.5,\axislabelheight){13}
\rput(8.0,\axislabelheight){14}
\rput(8.5,\axislabelheight){15}
\rput(9.0,\axislabelheight){16}
\rput(9.5,\axislabelheight){17}
\rput(10.0,\axislabelheight){18}
\rput(10.5,\axislabelheight){19}
\rput(11.0,\axislabelheight){20}
\rput(11.5,\axislabelheight){21}
\rput(12.0,\axislabelheight){22}
\rput(12.5,\axislabelheight){23}
\rput(13.0,\axislabelheight){24}
\rput(13.5,\axislabelheight){25}
\rput(14.0,\axislabelheight){26}
\rput(14.5,\axislabelheight){27}
\rput(15.0,\axislabelheight){28}
\rput(15.5,\axislabelheight){29}
\rput(16.0,\axislabelheight){30}
\rput(16.5,\axislabelheight){31}
}

\end{pspicture}
\caption{The occurrences of the subpatterns $P_1=\mathrm{A}$, $P_2=\mathrm{CC}$ and $P_3=\mathrm{GT}$ and the ranges they define in the text $T$ from Example~\ref{ex:runningexample}. Occurrences which are not relevant are crossed out. The bold occurrences of $P_3$ are the relevant occurrences of $P_k$ and their end positions 17,28 and 31 constitute the solution to the VLG problem. Consider the point in the execution of the algorithm  when  the occurrence $x$ of $P_2$ at position $\tau=26$  is reported by the Aho-Corasick automaton. At this time $L_2=\bigl[\; [17;20],[22;23],[25;26]\; \bigl]$ and $L_3=\bigl[\; [23,28] \; \bigl]$. The ranges $[17;20]$ and $[22;23]$ are now dead and are removed from $L_2$ in step~\ref{algstep:removedead}. %, since they end before $\tau - |P_i|=24$.
In step~\ref{algstep:relevancecheck} the algorithm determines that $x$ is relevant %because $startpox(x)=25\in [25;26]$. Then 
and $R(x)=[29;33]$ is appended to $L_3$: $L_3 = \bigl[\; [23;33] \; \bigl]$. 
\label{fig:runningexample}}
\end{figure}

\section{Analysis}\label{sec:analysis}
We now show that Algorithm~\ref{alg:ouralgorithm} solves the VLG problem in time $O(n\log k + m +\alpha)$ and space $O(m+A)$, implying Theorem~\ref{thm:vlgproblem}.

\subsection{Correctness}
To show that Algorithm~\ref{alg:ouralgorithm} finds exactly the relevant occurrences of $P_k$, we show by induction on $i$ that the algorithm in step~\ref{algstep:relevancecheck} correctly determines the relevancy of all occurrences of $P_i$, $i=1,2,\ldots,k$, in $T$.
 
\begin{description}
\item[Base case:] All occurrences of $P_1$ are by definition relevant and %clearly 
Algorithm~\ref{alg:ouralgorithm} correctly determines this in step~\ref{algstep:relevancecheck}.

\item[Inductive step:]
Let $y$ be an occurrence of $P_i$, $i>1$, that is reported at position $\tau$. There are two cases to consider.
\begin{enumerate}
\item $y$ is relevant.  By definition there is a relevant occurrence $x$ of $P_{i-1}$ in $T$, such that $\spos(y) = \tau - |P_i| \in R(x)$. By the induction hypothesis $x$ was correctly determined to be relevant by the algorithm. Since $\epos(x) < \tau$, $R(x)$ was appended to $L_{i}$ earlier in the execution of the algorithm. It remains to show that the range containing $\spos(y)$ is the first range in $L_i$ in step ~\ref{algstep:relevancecheck}.
When removing the dead ranges in $L_i$ in step~\ref{algstep:removedead}, all ranges $[a,b]$ where $b< \tau - |P_i|$ are  removed. Therefore the range containing $\tau -|P_i| = \spos(y)$ is the first range in $L_i$ after step~\ref{algstep:removedead}. It follows that the  algorithm correctly determines that $y$ is relevant.
\item $y$ is not relevant. Then there exists no relevant occurrence $x$ of $P_{i-1}$ such that $\spos(y) \in R(x)$. By the induction hypothesis there is no range in $L_i$ containing $\spos(y)$, since the algorithm only append ranges when a relevant occurrence is found. Consequently, the algorithm correctly determines that $y$ is not relevant.
\end{enumerate}
\ignore{

 If and only if such a relevant occurrence of $x$ exists in $T$, $L_i$ contains a range that contains $\spos(y)=\tau - |P_i|$, since the algorithm previously, by the induction hypothesis, would have appended $R(x)$ to $L_i$ if and only if $x$ was relevant. Thus $y$ is relevant iff $\spos(y)$ is contained in a range in $L_i$. It is safe to remove the dead ranges of $L_i$ in step~\ref{algstep:removedead}, since $\spos(y)=\tau-|P_i|$ can not be contained in any of these. Let $R^\prime$ be the first range in $L_i$ after the removal of dead ranges. If $\spos(y) \in R^\prime$ then $y$ is relevant. If $\spos(y) \notin R^\prime$ (or if $L_i$ is empty) then $\spos(y)$ also is not contained in any of the remaining ranges in $L_i$, since they are disjoint and sorted. Thus $y$ is relevant iff $\spos(y) \in R^\prime$ and hence the algorithm correctly determines the relevancy of $y$ in step~\ref{algstep:relevancecheck}.}
\end{description}

\subsection{Time and Space Complexity}\label{sec:complexityanalysis} 
The AC automaton for the subpatterns $P_1,P_2,\ldots,P_k$ can be built in time $O(m\log k)$ using $O(m)$ space, where $m=\sum_{i=1}^k |P_i|$. In the trivial case when $m > n$ we do not need to build the automaton. Hence, we will assume that $m \leq n$ in the following analysis. For each of the $\alpha$ occurrences of the strings $P_1, P_2, \ldots, P_k$ Algorithm~\ref{alg:ouralgorithm} first removes the dead ranges from $L_i$ and $L_{i+1}$ and performs a number of constant-time operations. Since both lists are sorted, the dead ranges can be removed by traversing the lists from the beginning. At most $\alpha$ ranges are ever added to the lists, and therefore the algorithm spends $O(\alpha)$ time in total on removing dead ranges. The total time is therefore $O((n+m)\log k+\alpha) = O(n\log k + m + \alpha)$.

To prove the space bound, we first show the following lemma.
\begin{lemma}\label{lem:listsizes}
At any time during the execution of the algorithm we have
\[
|L_i| \leq \left \lfloor \frac{2c_{i-1}+|P_i|+a_{i-1}}{c_{i-1}+1} \right \rfloor = O\left(\frac{|P_i|+a_{i-1}}{b_{i-1} - a_{i-1}+2} \right ),
\]
for $i=2,3,\ldots,k$, where $c_i=b_i-a_i+1$.
\end{lemma}

\begin{proof}
Consider list $L_i$ for some $i=2,\ldots,k$. Referring to Algorithm~\ref{alg:ouralgorithm}, the size of the list $L_i$ is only increased in step~\ref{algstep:appendrange}, when a range $R(x_j)$ defined by a relevant occurrence $x_j$ of $P_{i-1}$ is reported and $R(x_j)$ \emph{does not} adjoin or overlap the last range in $L_i$.

\begin{figure}[tb]
\centering
\newpsstyle{occStyle}{%
  fillstyle=vlines,hatchcolor=black!40,
  hatchwidth=0.1\pslinewidth,
  hatchsep=1\pslinewidth,
  linewidth=0.3pt
}

\newcommand{\textaxisheight}{3}
\psset{unit=0.5cm,yunit=0.8,arrowscale=1}
\begin{pspicture}[showgrid=false](1.2,1)(26,15)
\footnotesize
\psline[arrowscale=2.5,linewidth=0.8pt]{->}(1.2,\textaxisheight)(25.5,\textaxisheight)

\rput(1.25,4){\psframe[style=occStyle](3.50,1)}
{\psset{linewidth=0.3pt,linestyle=dashed,dash=1pt 1pt}
\psline{-}(3.75,5)(3.75,4)
\psline{-}(4.25,5)(4.25,4)
\rput(3,4.5){$P_{i-1}$}
\psline{-}(1.75,5)(1.75,4)
\psline{-}(2.25,5)(2.25,4)
}

\rput(11.25,10){\psframe[style=occStyle](6.50,1)}
{\psset{linewidth=0.3pt,linestyle=dashed,dash=1pt 1pt}
\psline{-}(11.75,11)(11.75,10)
\psline{-}(12.25,11)(12.25,10)
\psline{-}(12.75,11)(12.75,10)
\rput(14.5,10.5){$P_{i}$}
\psline{-}(16.25,11)(16.25,10)
\psline{-}(16.75,11)(16.75,10)
\psline{-}(17.25,11)(17.25,10)
}

\rput(14.25,8){\psframe[style=occStyle](3.50,1)}
{\psset{linewidth=0.3pt,linestyle=dashed,dash=1pt 1pt}
\psline{-}(14.75,9)(14.75,8)
\psline{-}(15.25,9)(15.25,8)
\rput(16,8.5){$P_{i-1}$}
\psline{-}(16.75,9)(16.75,8)
\psline{-}(17.25,9)(17.25,8)
}

%Intervals
\psline{|-|}(9.25,4.5)(11.75,4.5)\rput(10.5,5){$R(x_1)$}
\psline{|-|}(22.25,8.5)(24.75,8.5)\rput(23.5,9){$R(x_\ell)$}

{\psset{linestyle=dashed}

\psline{|-|}(12.25,5)(14.75,5)\rput(13.5,5.5){$R(x_2)$}
\psline{|-|}(15.25,6)(17.75,6)

\psline[linewidth=1pt,dash=1pt 2pt]{-}(17,6.5)(20.7,7)

\psline{|-|}(19.25,7.5)(21.75,7.5)

}

{\psset{linewidth=0.5pt}
\psline{|<->|}(11.75,14)(24.75,14)\rput*(18.25,14){$d$}
%\psline{|<->|}(9.25,13)(11.75,13)\rput*(10.5,13){$c_{i-1}$}
\psline{|<->|}(11.75,13)(17.75,13)\rput*(14.75,13){$|P_{i}|-1$}
\psline{|<->|}(17.75,13)(22.25,13)\rput*(20,13){$a_{i-1}$}
\psline{|<->|}(22.25,13)(24.75,13)\rput*(23.5,13){$c_{i-1}$}
\psline{|<->|}(17.75,12)(24.75,12)\rput*(20.6,12){$b_{i-1}+1$}

}

%Vertical lines
{\psset{linestyle=dashed,linewidth=0.3pt,dash=0.5pt 1pt}
\psline{-}(4.5,4)(4.5,3)
%\psline{-}(9.25,13)(9.25,4.5)
\psline{-}(11.75,14)(11.75,11)\psline{-}(11.75,10)(11.75,4.5)
\psline{-}(17.75,13)(17.75,11)\psline{-}(17.75,10)(17.75,9)\psline{-}(17.75,8)(17.75,3)
\psline{-}(22.25,13)(22.25,8.5)
\psline{-}(24.75,14)(24.75,8.5)

}

%Marks on T-axis
{
\psline{-}(4.5,2.8)(4.5,3.2)
\psline{-}(11.75,2.8)(11.75,3.2)
\psline{-}(12.25,2.8)(12.25,3.2)
\psline{-}(17.5,2.8)(17.5,3.2)

}

%Text on axis
{%\tiny
\rput[t](4.5,2.5){\parbox{2.9cm}{$x_1$ is reported and $R(x_1)$ is added to $L_i$}}
\rput[t](17.5,2.5){\parbox{2.8cm}{Last position where $R(x_1)$ is still alive}}
\rput{0}(12,2.3){$\underbrace{}_1$}
\rput(24,2.2){\scriptsize Position in $T$}
}
{\psset{linewidth=0.03,arrowscale=1.5}
\psline{->}(2,8)(3,5.3)\rput(1.9,8.35){$x_1$}
%\psline{->}(6,12)(11,10.8)\rput(5.7,12){$y$}
\psline{->}(8.1,7)(13.8,8.5)\rput(7.7,6.9){$x_\ell$}
}
\end{pspicture}

\caption{The worst-case situation where $\ell$, the maximum number of ranges are present in $L_i$. The figure only shows the first and the last occurrence of $P_{i-1}$ ($x_1$ and $x_\ell$) defining the $\ell$ ranges.\label{fig:listsizes}}
\end{figure}
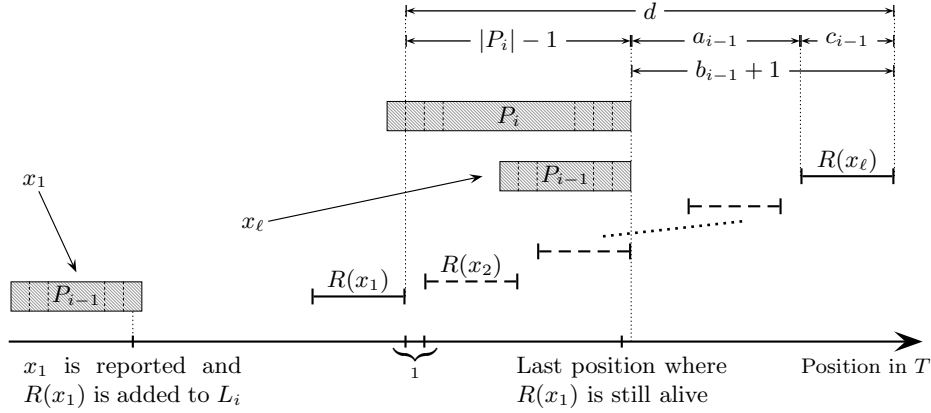

Let $R(x_1)=[s,t]$ be the first range in $L_i$ at an arbitrary time in the execution of the algorithm. We bound the number of additional ranges that can be added to $L_i$ from the time $R(x_1)$ became the first range in $L_i$ until $R(x_1)$ is removed. The last position where $R(x_1)$ is still alive is $\tau_a=t+|P_i|-1$. If a relevant occurrence $x_\ell$ of $P_{i-1}$ ends at this position, then the range $R(x_\ell)=[\tau_a+a_{i-1}+1;\tau_a+b_{i-1}+1]$ is appended to $L_i$. Hence, the maximum number of positions $d$ from $t$ to the end of $R(x_\ell)$ is
\begin{align*}
d&=\tau_a+b_{i-1}+1-t \\
 &=(t+|P_i|-1)+b_{i-1}+1-t \\
 &=|P_i|+b_{i-1} \\
 &=|P_i|+a_{i-1}+c_{i-1}-1 \enspace.
\end{align*}
In the worst case, all the ranges in $L_i$ are separated by exactly one position as illustrated in Fig.~\ref{fig:listsizes}. Therefore at most $\lfloor d / (c_{i-1}+1) \rfloor$ additional ranges can be added to $L_i$ before $R(x_1)$ is removed. Counting in $R(x_1)$ yields the following bound on the size of $L_i$
\[
|L_i| \leq \left \lfloor \frac{d}{c_{i-1}+1} \right \rfloor+1 = \left \lfloor \frac{2c_{i-1}+|P_i|+a_{i-1}}{c_{i-1}+1} \right \rfloor = O\left(\frac{|P_i|+a_{i-1}}{b_{i-1}-a_{i-1}+2}\right) \enspace .
\]
\qed

\end{proof}
By Lemma~\ref{lem:listsizes} the total number of ranges stored at any time during the processing of $T$ is at most
\[
O\left(\sum_{i=2}^k \frac{|P_i|+a_{i-1}}{b_{i-1}-a_{i-1}+2}\right)  = O\left(\;\sum_{i=1}^{k-1} \frac{|P_{i+1}|}{b_{i}-a_{i}+2}+\sum_{i=1}^{k-1} \frac{a_{i}}{{b_{i}-a_{i}+2}} \right ) = O\left( m + A \right)\enspace .
\]

Each range can be stored using $O(1)$ space, so this is an upper bound on the space needed to store the lists $L_2,\ldots,L_k$. The AC-automaton uses $O(m)$ space, so the total space required by our algorithm is $O(m+A)$.

In summary, the algorithm uses $O(n\log k + m + \alpha)$ time and $O(m + A)$ space. This completes the proof of Theorem~\ref{thm:vlgproblem}.

\section{Complete Characterization of Occurrences}
In this section we show how our algorithm can be extended to report not only the end position of $P_k$, but also the positions of $P_1,P_2,\ldots,P_{k-1}$ for each occurrence of $P$ in $T$.

The main idea is to construct a graph that encodes all occurrences of the VLG-pattern using $O(\alpha)$ space. For each occurrence of the VLG-pattern, the positions of the individual subpatterns can be reported by traversing this graph. This approach was also used by Rahman~et~al.~\cite{RILMS2006} and Morgante~et~al.~\cite{MPVZ2005}. We give a fast new algorithm for constructing this graph and show a black-box solution that can report the occurrences of the VLG-pattern without storing the complete graph.

We introduce the following simple definitions. If $P$ occurs in the text $T$, then a \emph{match combination} is a sequence $e_1, \ldots, e_k$ of end positions of $P_1, \ldots, P_k$ in $T$ corresponding to the match. The total number of match combinations of $P$ in $T$ is denoted $\beta$. Note that there can be many match combinations corresponding to a single match. See Fig.~\ref{fig:matchcombinations}.

\begin{figure}[H]
\centering
\newpsstyle{occStyle}{%
  fillstyle=vlines,hatchcolor=black!40,
  hatchwidth=0.1\pslinewidth,
  hatchsep=1\pslinewidth,
  linewidth=0.3pt
}

\psset{unit=0.5cm,yunit=0.75,xunit=1.5,arrowscale=1}
\begin{pspicture}[showgrid=false](1.2,2)(18,9)
\footnotesize
\psline[arrowscale=2.5,linewidth=0.8pt]{->}(1.2,3)(17.5,3)

%Text string
{\bfseries
\newcommand{\searchtextheight}{3.55}
{\color{remTextColor}
\rput(1.5,\searchtextheight){A}
\rput(2.0,\searchtextheight){T}
\rput(2.5,\searchtextheight){C}
\rput(3.0,\searchtextheight){G}
}
\rput(3.5,\searchtextheight){G}
\rput(4.0,\searchtextheight){C}
\rput(4.5,\searchtextheight){T}
\rput(5.0,\searchtextheight){C}
\rput(5.5,\searchtextheight){C}
\rput(6.0,\searchtextheight){A}
\rput(6.5,\searchtextheight){G}
\rput(7.0,\searchtextheight){A}
\rput(7.5,\searchtextheight){C}
\rput(8.0,\searchtextheight){C}
\rput(8.5,\searchtextheight){A}
\rput(9.0,\searchtextheight){G}
\rput(9.5,\searchtextheight){T}
{\color{remTextColor}
\rput(10.0,\searchtextheight){A}
\rput(10.5,\searchtextheight){C}
\rput(11.0,\searchtextheight){C}
\rput(11.5,\searchtextheight){C}
\rput(12.0,\searchtextheight){G}
\rput(12.5,\searchtextheight){T}
\rput(13.0,\searchtextheight){T}
\rput(13.5,\searchtextheight){C}
\rput(14.0,\searchtextheight){C}
\rput(14.5,\searchtextheight){G}
\rput(15.0,\searchtextheight){T}
\rput(15.5,\searchtextheight){G}
\rput(16.0,\searchtextheight){G}
\rput(16.5,\searchtextheight){T}
}
}

{%Axis labels
\tiny
\newcommand{\axislabelheight}{2.55}
\rput(1.5,\axislabelheight){1}
\rput(2.0,\axislabelheight){2}
\rput(2.5,\axislabelheight){3}
\rput(3.0,\axislabelheight){4}
\rput(3.5,\axislabelheight){5}
\rput(4.0,\axislabelheight){6}
\rput(4.5,\axislabelheight){7}
\rput(5.0,\axislabelheight){8}
\rput(5.5,\axislabelheight){9}
\rput(6.0,\axislabelheight){10}
\rput(6.5,\axislabelheight){11}
\rput(7.0,\axislabelheight){12}
\rput(7.5,\axislabelheight){13}
\rput(8.0,\axislabelheight){14}
\rput(8.5,\axislabelheight){15}
\rput(9.0,\axislabelheight){16}
\rput(9.5,\axislabelheight){17}
\rput(10.0,\axislabelheight){18}
\rput(10.5,\axislabelheight){19}
\rput(11.0,\axislabelheight){20}
\rput(11.5,\axislabelheight){21}
\rput(12.0,\axislabelheight){22}
\rput(12.5,\axislabelheight){23}
\rput(13.0,\axislabelheight){24}
\rput(13.5,\axislabelheight){25}
\rput(14.0,\axislabelheight){26}
\rput(14.5,\axislabelheight){27}
\rput(15.0,\axislabelheight){28}
\rput(15.5,\axislabelheight){29}
\rput(16.0,\axislabelheight){30}
\rput(16.5,\axislabelheight){31}
}

% Match combinations
{\bfseries
\rput(3.5,4.5){G}
\rput(4.0,4.5){C}
\rput(4.5,4.5){-}
\rput(5.0,4.5){-}
\rput(5.5,4.5){-}
\rput(6.0,4.5){A}
\rput(6.5,4.5){-}
\rput(7.0,4.5){-}
\rput(7.5,4.5){-}
\rput(8.0,4.5){-}
\rput(8.5,4.5){-}
\rput(9.0,4.5){-}
\rput(9.5,4.5){T}
}

{\bfseries
\rput(3.5,5.5){G}
\rput(4.0,5.5){C}
\rput(4.5,5.5){-}
\rput(5.0,5.5){-}
\rput(5.5,5.5){-}
\rput(6.0,5.5){-}
\rput(6.5,5.5){-}
\rput(7.0,5.5){A}
\rput(7.5,5.5){-}
\rput(8.0,5.5){-}
\rput(8.5,5.5){-}
\rput(9.0,5.5){-}
\rput(9.5,5.5){T}
}

{\bfseries
\rput(3.5,6.5){G}
\rput(4.0,6.5){-}
\rput(4.5,6.5){-}
\rput(5.0,6.5){C}
\rput(5.5,6.5){-}
\rput(6.0,6.5){A}
\rput(6.5,6.5){-}
\rput(7.0,6.5){-}
\rput(7.5,6.5){-}
\rput(8.0,6.5){-}
\rput(8.5,6.5){-}
\rput(9.0,6.5){-}
\rput(9.5,6.5){T}
}

{\bfseries
\rput(3.5,7.5){G}
\rput(4.0,7.5){-}
\rput(4.5,7.5){-}
\rput(5.0,7.5){C}
\rput(5.5,7.5){-}
\rput(6.0,7.5){-}
\rput(6.5,7.5){-}
\rput(7.0,7.5){A}
\rput(7.5,7.5){-}
\rput(8.0,7.5){-}
\rput(8.5,7.5){-}
\rput(9.0,7.5){-}
\rput(9.5,7.5){T}
}

{\bfseries
\rput(3.5,8.5){G}
\rput(4.0,8.5){-}
\rput(4.5,8.5){-}
\rput(5.0,8.5){-}
\rput(5.5,8.5){C}
\rput(6.0,8.5){-}
\rput(6.5,8.5){-}
\rput(7.0,8.5){A}
\rput(7.5,8.5){-}
\rput(8.0,8.5){-}
\rput(8.5,8.5){-}
\rput(9.0,8.5){-}
\rput(9.5,8.5){T}
}

\psbrace[ref=l](9.5,4)(9.5,9){\scalefont{0.85}{The five match combinations.}}

\end{pspicture}
\caption{The text sequence is the same as in the previous examples. The substring $S$ from position 5 to 17 (highlighted in bold) matches the VLG-pattern $Q=G \cdot \gap{0}{3} \cdot C \cdot \gap{1}{6} \cdot A \cdot \gap{2}{7} \cdot T$. As the figure shows, this match contains the following five match combinations: [5,9,12,17],[5,8,12,17],[5,8,10,17],[5,6,12,17],[5,6,10,17].\label{fig:matchcombinations}}
\end{figure}

\noindent Due to the inequality of arithmetic and geometric means, the total number of match combinations $\beta$ is maximized, when the $\alpha$ occurrences are distributed evenly over $P_1,P_2,\ldots,P_k$ and each occurrence of $P_{i}$ is compatible to all occurrences of $P_{i-1}$ for $i=2,\ldots,k$. So in the worst case $\beta = \Theta \left((\frac{\alpha}{k})^k \right)$, which is exponential in the number of gaps. All these match combinations can be encoded in a directed graph using $O(\frac{\alpha^2}{k})$ space as follows. The nodes in the graph are the relevant occurrences of $P_1,P_2,\ldots,P_k$ in $T$. Two nodes $x$ of $P_{i-1}$ and $y$ of $P_{i}$ are connected by an edge from $y$ to $
x$ if and only if $startpos(y) \in R(x)$. In that case we also say that $x$ and $y$ are \emph{compatible}. We denote this graph as the \emph{gap graph} for $P$ and $T$. See Fig.~\ref{fig:explicitgapgraph}. Since the number of nodes in the gap graph is at most $\alpha$, and there are $O\left((\frac{\alpha}{k})^2\right)$ edges between the $k$ layers in the worst case, we can store the graph using $O(\frac{\alpha^2}{k})$ space.

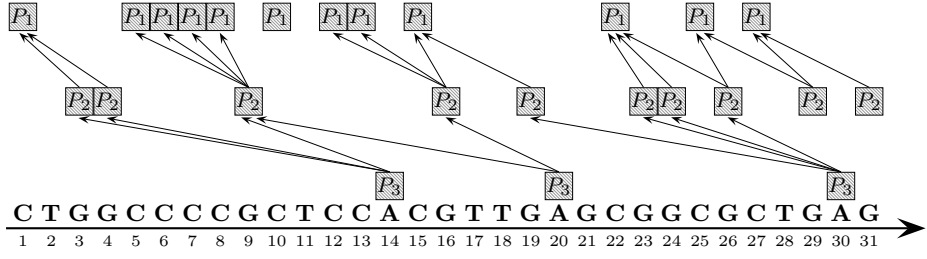
\begin{figure}[H]
\centering
\newpsstyle{occStyle}{%
  fillstyle=vlines,hatchcolor=black!40,
  hatchwidth=0.1\pslinewidth,
  hatchsep=1\pslinewidth,
  linewidth=0.3pt
}

\psset{unit=0.5cm,yunit=0.75,xunit=1.5,arrowscale=1}
\begin{pspicture}[showgrid=false](1.2,2)(18,11)
\footnotesize
\psline[arrowscale=2.5,linewidth=0.8pt]{->}(1.2,3)(17.5,3)

%Text string
{\bfseries
\newcommand{\searchtextheight}{3.55}
\rput(1.5,\searchtextheight){C}
\rput(2.0,\searchtextheight){T}
\rput(2.5,\searchtextheight){G}
\rput(3.0,\searchtextheight){G}
\rput(3.5,\searchtextheight){C}
\rput(4.0,\searchtextheight){C}
\rput(4.5,\searchtextheight){C}
\rput(5.0,\searchtextheight){C}
\rput(5.5,\searchtextheight){G}
\rput(6.0,\searchtextheight){C}
\rput(6.5,\searchtextheight){T}
\rput(7.0,\searchtextheight){C}
\rput(7.5,\searchtextheight){C}
\rput(8.0,\searchtextheight){A}
\rput(8.5,\searchtextheight){C}
\rput(9.0,\searchtextheight){G}
\rput(9.5,\searchtextheight){T}
\rput(10.0,\searchtextheight){T}
\rput(10.5,\searchtextheight){G}
\rput(11.0,\searchtextheight){A}
\rput(11.5,\searchtextheight){G}
\rput(12.0,\searchtextheight){C}
\rput(12.5,\searchtextheight){G}
\rput(13.0,\searchtextheight){G}
\rput(13.5,\searchtextheight){C}
\rput(14.0,\searchtextheight){G}
\rput(14.5,\searchtextheight){C}
\rput(15.0,\searchtextheight){T}
\rput(15.5,\searchtextheight){G}
\rput(16.0,\searchtextheight){A}
\rput(16.5,\searchtextheight){G}
}

{%Axis labels
\tiny
\newcommand{\axislabelheight}{2.55}
\rput(1.5,\axislabelheight){1}
\rput(2.0,\axislabelheight){2}
\rput(2.5,\axislabelheight){3}
\rput(3.0,\axislabelheight){4}
\rput(3.5,\axislabelheight){5}
\rput(4.0,\axislabelheight){6}
\rput(4.5,\axislabelheight){7}
\rput(5.0,\axislabelheight){8}
\rput(5.5,\axislabelheight){9}
\rput(6.0,\axislabelheight){10}
\rput(6.5,\axislabelheight){11}
\rput(7.0,\axislabelheight){12}
\rput(7.5,\axislabelheight){13}
\rput(8.0,\axislabelheight){14}
\rput(8.5,\axislabelheight){15}
\rput(9.0,\axislabelheight){16}
\rput(9.5,\axislabelheight){17}
\rput(10.0,\axislabelheight){18}
\rput(10.5,\axislabelheight){19}
\rput(11.0,\axislabelheight){20}
\rput(11.5,\axislabelheight){21}
\rput(12.0,\axislabelheight){22}
\rput(12.5,\axislabelheight){23}
\rput(13.0,\axislabelheight){24}
\rput(13.5,\axislabelheight){25}
\rput(14.0,\axislabelheight){26}
\rput(14.5,\axislabelheight){27}
\rput(15.0,\axislabelheight){28}
\rput(15.5,\axislabelheight){29}
\rput(16.0,\axislabelheight){30}
\rput(16.5,\axislabelheight){31}
}

{\scriptsize

%P1's
\rput(14.25,10){\psframe[style=occStyle](0.5,1)}
\rput(14.5,10.5){$P_1$}
\rput(14.5,10){\rnode{p112}{}}

\rput(13.25,10){\psframe[style=occStyle](0.5,1)}
\rput(13.5,10.5){$P_1$}
\rput(13.5,10){\rnode{p111}{}}

\rput(11.75,10){\psframe[style=occStyle](0.5,1)}
\rput(12,10.5){$P_1$}
\rput(12,10){\rnode{p110}{}}

\rput(8.25,10){\psframe[style=occStyle](0.5,1)}
\rput(8.5,10.5){$P_1$}
\rput(8.5,10){\rnode{p19}{}}

\rput(7.25,10){\psframe[style=occStyle](0.5,1)}
\rput(7.5,10.5){$P_1$}
\rput(7.5,10){\rnode{p18}{}}

\rput(6.75,10){\psframe[style=occStyle](0.5,1)}
\rput(7,10.5){$P_1$}
\rput(7,10){\rnode{p17}{}}

\rput(5.75,10){\psframe[style=occStyle](0.5,1)}
\rput(6,10.5){$P_1$}
\rput(6,10){\rnode{p16}{}}

\rput(4.75,10){\psframe[style=occStyle](0.5,1)}
\rput(5,10.5){$P_1$}
\rput(5,10){\rnode{p15}{}}

\rput(4.25,10){\psframe[style=occStyle](0.5,1)}
\rput(4.5,10.5){$P_1$}
\rput(4.5,10){\rnode{p14}{}}

\rput(3.75,10){\psframe[style=occStyle](0.5,1)}
\rput(4,10.5){$P_1$}
\rput(4,10){\rnode{p13}{}}

\rput(3.25,10){\psframe[style=occStyle](0.5,1)}
\rput(3.5,10.5){$P_1$}
\rput(3.5,10){\rnode{p12}{}}

\rput(1.25,10){\psframe[style=occStyle](0.5,1)}
\rput(1.5,10.5){$P_1$}
\rput(1.5,10){\rnode{p11}{}}

%P2's
\rput(16.25,7){\psframe[style=occStyle](0.5,1)}
\rput(16.5,7.5){$P_2$}
\rput(16.5,7){\rnode{p210b}{}}
\rput(16.5,8){\rnode{p210t}{}}

\rput(15.25,7){\psframe[style=occStyle](0.5,1)}
\rput(15.5,7.5){$P_2$}
\rput(15.5,7){\rnode{p29b}{}}
\rput(15.5,8){\rnode{p29t}{}}

\rput(13.75,7){\psframe[style=occStyle](0.5,1)}
\rput(14,7.5){$P_2$}
\rput(14,7){\rnode{p28b}{}}
\rput(14,8){\rnode{p28t}{}}

\rput(12.75,7){\psframe[style=occStyle](0.5,1)}
\rput(13,7.5){$P_2$}
\rput(13,7){\rnode{p27b}{}}
\rput(13,8){\rnode{p27t}{}}

\rput(12.25,7){\psframe[style=occStyle](0.5,1)}
\rput(12.5,7.5){$P_2$}
\rput(12.5,7){\rnode{p26b}{}}
\rput(12.5,8){\rnode{p26t}{}}

\rput(10.25,7){\psframe[style=occStyle](0.5,1)}
\rput(10.5,7.5){$P_2$}
\rput(10.5,7){\rnode{p25b}{}}
\rput(10.5,8){\rnode{p25t}{}}

\rput(8.75,7){\psframe[style=occStyle](0.5,1)}
\rput(9,7.5){$P_2$}
\rput(9,7){\rnode{p24b}{}}
\rput(9,8){\rnode{p24t}{}}

\rput(5.25,7){\psframe[style=occStyle](0.5,1)}
\rput(5.5,7.5){$P_2$}
\rput(5.5,7){\rnode{p23b}{}}
\rput(5.5,8){\rnode{p23t}{}}

\rput(2.75,7){\psframe[style=occStyle](0.5,1)}
\rput(3,7.5){$P_2$}
\rput(3,7){\rnode{p22b}{}}
\rput(3,8){\rnode{p22t}{}}

\rput(2.25,7){\psframe[style=occStyle](0.5,1)}
\rput(2.5,7.5){$P_2$}
\rput(2.5,7){\rnode{p21b}{}}
\rput(2.5,8){\rnode{p21t}{}}

%P3's
\rput(15.75,4){\psframe[style=occStyle](0.5,1)}
\rput(16,4.5){$P_3$}
\rput(16,5){\rnode{p33}{}}

\rput(10.75,4){\psframe[style=occStyle](0.5,1)}
\rput(11,4.5){$P_3$}
\rput(11,5){\rnode{p32}{}}

\rput(7.75,4){\psframe[style=occStyle](0.5,1)}
\rput(8,4.5){$P_3$}
\rput(8,5){\rnode{p31}{}}
}

%Edges
{\psset{linewidth=0.3pt,arrowscale=1.4,angleB=90,angleA=-20,offsetB=-0.05}
 \SpecialCoor

%LAYER 1
 \psline{->}(p21t)([nodesep=-0.1,angle=0,offset=-0.05]p11)
 \psline{->}(p22t)([nodesep=0.1,angle=0,offset=-0.05]p11)

 \psline{->}(p23t)([offset=-0.05]p12)
 \psline{->}(p23t)([offset=-0.05]p13)
 \psline{->}(p23t)([offset=-0.05]p14)
 \psline{->}(p23t)([offset=-0.05]p15)

 \psline{->}(p24t)([offset=-0.05]p17)
 \psline{->}(p24t)([offset=-0.05]p18)
 \psline{->}(p24t)([nodesep=-0.1,angle=0,offset=-0.05]p19)

 \psline{->}(p25t)([nodesep=0.1,angle=0,offset=-0.05]p19)

 \psline{->}(p26t)([nodesep=-0.2,angle=0,offset=-0.05]p110)

 \psline{->}(p27t)([nodesep=0,angle=0,offset=-0.05]p110)

 \psline{->}(p28t)([nodesep=0.2,angle=0,offset=-0.05]p110)
 \psline{->}(p28t)([nodesep=-0.1,angle=0,offset=-0.05]p111)

 \psline{->}(p29t)([nodesep=0.1,angle=0,offset=-0.05]p111)
 \psline{->}(p29t)([nodesep=-0.1,angle=0,offset=-0.05]p112)

 \psline{->}(p210t)([nodesep=0.1,angle=0,offset=-0.05]p112)

%LAYER 2
 \psline{->}(p31)([offset=-0.07]p21b)
 \psline{->}(p31)([offset=-0.07]p22b)
 \psline{->}(p31)([nodesep=-0.2,angle=0,offset=-0.05]p23b)

 \psline{->}(p32)([nodesep=0.2,angle=0,offset=-0.07]p23b)
 \psline{->}(p32)([offset=-0.05]p24b)

 \psline{->}(p33)([offset=-0.07]p25b)
 \psline{->}(p33)([offset=-0.05]p26b)
 \psline{->}(p33)([offset=-0.05]p27b)
 \psline{->}(p33)([offset=-0.05]p28b)

% \ncline{->}{p11}{p23i}
% \ncline{->}{p11}{p24i}

% \ncline{->}{p12}{p25i}
% %\ncline{->}{p12}{p26i}
% \psline{->}(p12)([nodesep=-0.1,angle=0,offset=0.05]p26i)

% \psline{->}(p13)([nodesep=0.2,angle=0,offset=0.05]p26i)
% %\ncline{->}{p13}{p26i}
% \ncline{->}{p13}{p27i}
% \ncline{->}{p13}{p28i}

% \psline{->}(p24o)([nodesep=-0.1,angle=0,offset=0.05]p32)
% \psline{->}(p25o)([nodesep=0.1,angle=0,offset=0.05]p32)
% \ncline{->}{p25o}{p33}
% \ncline{->}{p25o}{p34}

% \ncline{->}{p26o}{p35}
% \ncline{->}{p26o}{p36}
% \ncline{->}{p26o}{p37}
% \ncline{->}{p26o}{p38}

% \psline{->}(p27o)([nodesep=-0.1,angle=0,offset=0.05]p39)
% \psline{->}(p28o)([nodesep=0.1,angle=0,offset=0.05]p39)

 }
\end{pspicture}

\caption{The gap graph for the VLG-pattern $R=C \cdot \gap{0}{3} \cdot G\cdot \gap{3}{10} \cdot A$ and the text $T=\textrm{CTGGCCCCGCTCCACGTTGAGCGGCGCTGAG}$.\label{fig:explicitgapgraph}}
\end{figure}
\noindent If the $j$ occurrences $x_1,x_2,\ldots,x_j$ of $P_{i}$ (appearing in that order in $T$) are all compatible with the same occurrence $y$ of $P_{i+1}$, then the $j$ edges $(y,x_1),(y,x_2),\ldots,(y,x_j)$ are all present in the gap graph. Due to the following lemma, the edges $(y,x_2),\ldots,(y,x_{j-1})$ are redundant.

\begin{lemma}
Let $x_1$ and $x_2$ be two occurrences of $P_i$, $i=1,\ldots,k-1$, both compatible with the same occurrence $y$ of $P_{i+1}$. Assume without loss of generality that $startpos(x_1) < startpos(x_2)$ and let $x^\prime$ be another occurrence of $P_{i}$ such that $startpos(x_1) \leq startpos(x^\prime) \leq startpos(x_2)$, then $x^\prime$ is also compatible with $y$.
\end{lemma}
\begin{proof}
Since $startpos(y) \in R(x_1)$ and $startpos(y) \in R(x_2)$, we have that $startpos(y) \in R(x_1) \cap R(x_2)$. Furthermore since $startpos(x_1) \leq startpos(x^\prime) \leq startpos(x_2)$, it holds that $R(x_1) \cap R(x_2) \subseteq R(x^\prime)$, so $startpos(y) \in R(x^\prime)$.\qed
\end{proof}

\noindent Leaving out the redundant edges in the gap graph, we get a new graph, which we denote the \emph{implicit gap graph}. For an example, see Fig.~\ref{fig:implicitgapgraph}. In this graph the out-degree of each node is at most two, so the number of edges is now linear in the number of nodes, and consequently we can store the implicit gap graph using $O(\alpha)$ space.
\begin{figure}[H]
\centering
\newpsstyle{occStyle}{%
  fillstyle=vlines,hatchcolor=black!40,
  hatchwidth=0.1\pslinewidth,
  hatchsep=1\pslinewidth,
  linewidth=0.3pt
}

\psset{unit=0.5cm,yunit=0.75,xunit=1.5,arrowscale=1}
\begin{pspicture}[showgrid=false](1.2,2)(18,11)
\footnotesize
\psline[arrowscale=2.5,linewidth=0.8pt]{->}(1.2,3)(17.5,3)

%Text string
{\bfseries
\newcommand{\searchtextheight}{3.55}
\rput(1.5,\searchtextheight){C}
\rput(2.0,\searchtextheight){T}
\rput(2.5,\searchtextheight){G}
\rput(3.0,\searchtextheight){G}
\rput(3.5,\searchtextheight){C}
\rput(4.0,\searchtextheight){C}
\rput(4.5,\searchtextheight){C}
\rput(5.0,\searchtextheight){C}
\rput(5.5,\searchtextheight){G}
\rput(6.0,\searchtextheight){C}
\rput(6.5,\searchtextheight){T}
\rput(7.0,\searchtextheight){C}
\rput(7.5,\searchtextheight){C}
\rput(8.0,\searchtextheight){A}
\rput(8.5,\searchtextheight){C}
\rput(9.0,\searchtextheight){G}
\rput(9.5,\searchtextheight){T}
\rput(10.0,\searchtextheight){T}
\rput(10.5,\searchtextheight){G}
\rput(11.0,\searchtextheight){A}
\rput(11.5,\searchtextheight){G}
\rput(12.0,\searchtextheight){C}
\rput(12.5,\searchtextheight){G}
\rput(13.0,\searchtextheight){G}
\rput(13.5,\searchtextheight){C}
\rput(14.0,\searchtextheight){G}
\rput(14.5,\searchtextheight){C}
\rput(15.0,\searchtextheight){T}
\rput(15.5,\searchtextheight){G}
\rput(16.0,\searchtextheight){A}
\rput(16.5,\searchtextheight){G}
}

{%Axis labels
\tiny
\newcommand{\axislabelheight}{2.55}
\rput(1.5,\axislabelheight){1}
\rput(2.0,\axislabelheight){2}
\rput(2.5,\axislabelheight){3}
\rput(3.0,\axislabelheight){4}
\rput(3.5,\axislabelheight){5}
\rput(4.0,\axislabelheight){6}
\rput(4.5,\axislabelheight){7}
\rput(5.0,\axislabelheight){8}
\rput(5.5,\axislabelheight){9}
\rput(6.0,\axislabelheight){10}
\rput(6.5,\axislabelheight){11}
\rput(7.0,\axislabelheight){12}
\rput(7.5,\axislabelheight){13}
\rput(8.0,\axislabelheight){14}
\rput(8.5,\axislabelheight){15}
\rput(9.0,\axislabelheight){16}
\rput(9.5,\axislabelheight){17}
\rput(10.0,\axislabelheight){18}
\rput(10.5,\axislabelheight){19}
\rput(11.0,\axislabelheight){20}
\rput(11.5,\axislabelheight){21}
\rput(12.0,\axislabelheight){22}
\rput(12.5,\axislabelheight){23}
\rput(13.0,\axislabelheight){24}
\rput(13.5,\axislabelheight){25}
\rput(14.0,\axislabelheight){26}
\rput(14.5,\axislabelheight){27}
\rput(15.0,\axislabelheight){28}
\rput(15.5,\axislabelheight){29}
\rput(16.0,\axislabelheight){30}
\rput(16.5,\axislabelheight){31}
}

{\scriptsize

%P1's
\rput(14.25,10){\psframe[style=occStyle](0.5,1)}
\rput(14.5,10.5){$P_1$}
\rput(14.5,10){\rnode{p112}{}}

\rput(13.25,10){\psframe[style=occStyle](0.5,1)}
\rput(13.5,10.5){$P_1$}
\rput(13.5,10){\rnode{p111}{}}

\rput(11.75,10){\psframe[style=occStyle](0.5,1)}
\rput(12,10.5){$P_1$}
\rput(12,10){\rnode{p110}{}}

\rput(8.25,10){\psframe[style=occStyle](0.5,1)}
\rput(8.5,10.5){$P_1$}
\rput(8.5,10){\rnode{p19}{}}

\rput(7.25,10){\psframe[style=occStyle](0.5,1)}
\rput(7.5,10.5){$P_1$}
\rput(7.5,10){\rnode{p18}{}}

\rput(6.75,10){\psframe[style=occStyle](0.5,1)}
\rput(7,10.5){$P_1$}
\rput(7,10){\rnode{p17}{}}

\rput(5.75,10){\psframe[style=occStyle](0.5,1)}
\rput(6,10.5){$P_1$}
\rput(6,10){\rnode{p16}{}}

\rput(4.75,10){\psframe[style=occStyle](0.5,1)}
\rput(5,10.5){$P_1$}
\rput(5,10){\rnode{p15}{}}

\rput(4.25,10){\psframe[style=occStyle](0.5,1)}
\rput(4.5,10.5){$P_1$}
\rput(4.5,10){\rnode{p14}{}}

\rput(3.75,10){\psframe[style=occStyle](0.5,1)}
\rput(4,10.5){$P_1$}
\rput(4,10){\rnode{p13}{}}

\rput(3.25,10){\psframe[style=occStyle](0.5,1)}
\rput(3.5,10.5){$P_1$}
\rput(3.5,10){\rnode{p12}{}}

\rput(1.25,10){\psframe[style=occStyle](0.5,1)}
\rput(1.5,10.5){$P_1$}
\rput(1.5,10){\rnode{p11}{}}

%P2's
\rput(16.25,7){\psframe[style=occStyle](0.5,1)}
\rput(16.5,7.5){$P_2$}
\rput(16.5,7){\rnode{p210b}{}}
\rput(16.5,8){\rnode{p210t}{}}

\rput(15.25,7){\psframe[style=occStyle](0.5,1)}
\rput(15.5,7.5){$P_2$}
\rput(15.5,7){\rnode{p29b}{}}
\rput(15.5,8){\rnode{p29t}{}}

\rput(13.75,7){\psframe[style=occStyle](0.5,1)}
\rput(14,7.5){$P_2$}
\rput(14,7){\rnode{p28b}{}}
\rput(14,8){\rnode{p28t}{}}

\rput(12.75,7){\psframe[style=occStyle](0.5,1)}
\rput(13,7.5){$P_2$}
\rput(13,7){\rnode{p27b}{}}
\rput(13,8){\rnode{p27t}{}}

\rput(12.25,7){\psframe[style=occStyle](0.5,1)}
\rput(12.5,7.5){$P_2$}
\rput(12.5,7){\rnode{p26b}{}}
\rput(12.5,8){\rnode{p26t}{}}

\rput(10.25,7){\psframe[style=occStyle](0.5,1)}
\rput(10.5,7.5){$P_2$}
\rput(10.5,7){\rnode{p25b}{}}
\rput(10.5,8){\rnode{p25t}{}}

\rput(8.75,7){\psframe[style=occStyle](0.5,1)}
\rput(9,7.5){$P_2$}
\rput(9,7){\rnode{p24b}{}}
\rput(9,8){\rnode{p24t}{}}

\rput(5.25,7){\psframe[style=occStyle](0.5,1)}
\rput(5.5,7.5){$P_2$}
\rput(5.5,7){\rnode{p23b}{}}
\rput(5.5,8){\rnode{p23t}{}}

\rput(2.75,7){\psframe[style=occStyle](0.5,1)}
\rput(3,7.5){$P_2$}
\rput(3,7){\rnode{p22b}{}}
\rput(3,8){\rnode{p22t}{}}

\rput(2.25,7){\psframe[style=occStyle](0.5,1)}
\rput(2.5,7.5){$P_2$}
\rput(2.5,7){\rnode{p21b}{}}
\rput(2.5,8){\rnode{p21t}{}}

%P3's
\rput(15.75,4){\psframe[style=occStyle](0.5,1)}
\rput(16,4.5){$P_3$}
\rput(16,5){\rnode{p33}{}}

\rput(10.75,4){\psframe[style=occStyle](0.5,1)}
\rput(11,4.5){$P_3$}
\rput(11,5){\rnode{p32}{}}

\rput(7.75,4){\psframe[style=occStyle](0.5,1)}
\rput(8,4.5){$P_3$}
\rput(8,5){\rnode{p31}{}}
}

%Edges
{\psset{linewidth=0.3pt,arrowscale=1.4,angleB=90,angleA=-20,offsetB=-0.05}
 \SpecialCoor

%LAYER 1
 \psline{->}(p21t)([nodesep=-0.1,angle=0,offset=-0.05]p11)
 \psline{->}(p22t)([nodesep=0.1,angle=0,offset=-0.05]p11)

 \psline{->}(p23t)([offset=-0.05]p12)
% \psline{->}(p23t)([offset=-0.05]p13)
% \psline{->}(p23t)([offset=-0.05]p14)
 \psline{->}(p23t)([offset=-0.05]p15)

 \psline{->}(p24t)([offset=-0.05]p17)
% \psline{->}(p24t)([offset=-0.05]p18)
 \psline{->}(p24t)([nodesep=-0.1,angle=0,offset=-0.05]p19)

 \psline{->}(p25t)([nodesep=0.1,angle=0,offset=-0.05]p19)

 \psline{->}(p26t)([nodesep=-0.2,angle=0,offset=-0.05]p110)

 \psline{->}(p27t)([nodesep=0,angle=0,offset=-0.05]p110)

 \psline{->}(p28t)([nodesep=0.2,angle=0,offset=-0.05]p110)
 \psline{->}(p28t)([nodesep=-0.1,angle=0,offset=-0.05]p111)

 \psline{->}(p29t)([nodesep=0.1,angle=0,offset=-0.05]p111)
 \psline{->}(p29t)([nodesep=-0.1,angle=0,offset=-0.05]p112)

 \psline{->}(p210t)([nodesep=0.1,angle=0,offset=-0.05]p112)

%LAYER 2
 \psline{->}(p31)([offset=-0.07]p21b)
% \psline{->}(p31)([offset=-0.07]p22b)
 \psline{->}(p31)([nodesep=-0.2,angle=0,offset=-0.05]p23b)

 \psline{->}(p32)([nodesep=0.2,angle=0,offset=-0.07]p23b)
 \psline{->}(p32)([offset=-0.05]p24b)

 \psline{->}(p33)([offset=-0.07]p25b)
% \psline{->}(p33)([offset=-0.05]p26b)
% \psline{->}(p33)([offset=-0.05]p27b)
 \psline{->}(p33)([offset=-0.05]p28b)

 }
\end{pspicture}

\caption{The implicit gap graph for the VLG-pattern $R=C \cdot \gap{0}{3} \cdot G\cdot \gap{3}{10} \cdot A$ and the text $T=\textrm{CTGGCCCCGCTCCACGTTGAGCGGCGCTGAG}$. The out-degree of each node is at most two. Compare to Fig.~\ref{fig:explicitgapgraph}.\label{fig:implicitgapgraph}}
\end{figure}
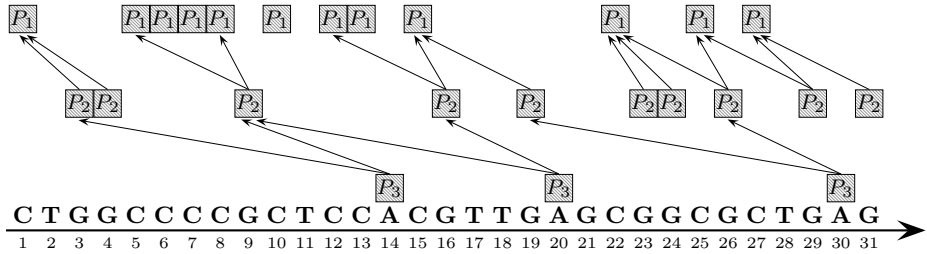

\noindent In the context of these new definitions, we are interested in solving the two following problems:

\begin{description}
\item [The \emph{reporting variable length gaps problem}] (RVLG problem)  is to output all match combinations of $P$ in $T$.

\item [The \emph{implicit reporting variable length gaps problem}] (IRVLG problem) is to output the implicit gap graph of all match combinations of $P$ in $T$. 
\end{description}

\subsection{Constructing the Implicit Gap Graph}\label{sec:constructing_implicitgapgraph}
Algorithm~\ref{alg:irvlgalgorithm} describes how to build the implicit gap graph. Recall that in Algorithm~\ref{alg:ouralgorithm} the ranges in $L_i$ allowed us to determine the relevancy of a newly reported occurrence $x$ of $P_i$ by inspecting the first range in $L_i$ (after the dead ranges had been removed). To build the implicit gap graph, we need to not only determine the relevancy of $x$, but also the first and last occurrence of $P_{i-1}$ compatible with $x$. This information allows us to add the correct edges to the implicit gap graph.

\begin{figure}
\centering
\newpsstyle{occStyle}{%
  fillstyle=vlines,hatchcolor=black!40,
  hatchwidth=0.1\pslinewidth,
  hatchsep=1\pslinewidth,
  linewidth=0.3pt
}

\psset{unit=0.5cm,yunit=0.75,xunit=1.5,arrowscale=1}
\begin{pspicture}[showgrid=false](1.2,-2)(12.5,10)
\footnotesize
\psline[arrowscale=2.5,linewidth=0.8pt]{->}(1.2,3)(12.5,3)\rput(11.3,1.75){\scriptsize Position in $T$}
\psline[arrowscale=2.5,linewidth=0.8pt]{->}(1.2,3)(1.2,10)\rput(0.5,8.5){\scriptsize Time}

%Text string
{\bfseries
\newcommand{\searchtextheight}{3.55}
\rput(1.5,\searchtextheight){G}
\rput(2.0,\searchtextheight){A}
\rput(2.5,\searchtextheight){C}
\rput(3.0,\searchtextheight){A}
\rput(3.5,\searchtextheight){C}
\rput(4.0,\searchtextheight){A}
\rput(4.5,\searchtextheight){C}
\rput(5.0,\searchtextheight){C}
\rput(5.5,\searchtextheight){T}
\rput(6.0,\searchtextheight){G}
\rput(6.5,\searchtextheight){G}
\rput(7.0,\searchtextheight){C}
\rput(7.5,\searchtextheight){A}
\rput(8.0,\searchtextheight){T}
\rput(8.5,\searchtextheight){A}
\rput(9.0,\searchtextheight){G}
\rput(9.5,\searchtextheight){C}
\rput(10.0,\searchtextheight){C}
\rput(10.5,\searchtextheight){G}
\rput(11.0,\searchtextheight){A}
}

{%Axis labels
\tiny
\newcommand{\axislabelheight}{2.55}
\rput(1.5,\axislabelheight){1}
\rput(2.0,\axislabelheight){2}
\rput(2.5,\axislabelheight){3}
\rput(3.0,\axislabelheight){4}
\rput(3.5,\axislabelheight){5}
\rput(4.0,\axislabelheight){6}
\rput(4.5,\axislabelheight){7}
\rput(5.0,\axislabelheight){8}
\rput(5.5,\axislabelheight){9}
\rput(6.0,\axislabelheight){10}
\rput(6.5,\axislabelheight){11}
\rput(7.0,\axislabelheight){12}
\rput(7.5,\axislabelheight){13}
\rput(8.0,\axislabelheight){14}
\rput(8.5,\axislabelheight){15}
\rput(9.0,\axislabelheight){16}
\rput(9.5,\axislabelheight){17}
\rput(10.0,\axislabelheight){18}
\rput(10.5,\axislabelheight){19}
\rput(11.0,\axislabelheight){20}
}

{\scriptsize
\rput(1.5,7.5){$x_1$}\psline[linewidth=0.03,arrowscale=1.5]{->}(1.6,7.15)(2.25,5.6)
\rput(1.75,4.5){\psframe[style=occStyle](1,1)}
\rput(2.25,5){$P_1$}
\psline{|-|}(3.25,5)(5.75,5)

\rput(2.5,9){$x_2$}\psline[linewidth=0.03,arrowscale=1.5]{->}(2.6,8.65)(3.25,7.1)
\rput(2.75,6){\psframe[style=occStyle](1,1)}
\rput(3.25,6.5){$P_1$}
\psline{|-|}(4.25,6.5)(6.75,6.5)

\rput(3.5,10.5){$x_3$}\psline[linewidth=0.03,arrowscale=1.5]{->}(3.6,10.15)(4.25,8.6)
\rput(3.75,7.5){\psframe[style=occStyle](1,1)}
\rput(4.25,8){$P_1$}
\psline{|-|}(5.25,8)(7.75,8)

\rput(5,12){$y$}\psline[linewidth=0.03,arrowscale=1.5]{->}(5.1,11.65)(5.50,10.1)
\rput(5.25,9){\psframe[style=occStyle](0.5,1)}
\rput(5.50,9.5){$P_2$}

\rput(7.75,10){\psframe[linewidth=0.3\pslinewidth](0.5,1)}
\psline[linewidth=0.3\pslinewidth]{-}(7.75,10)(8.25,11)
\psline[linewidth=0.3\pslinewidth]{-}(7.75,11)(8.25,10)
\rput(8,10.5){$P_2$}
}

{\normalsize
\rput(1.6,1){$L_2^f:$}
\psline{|-|}(3.25,1)(5.74,1)\rput(4.5,1.35){\scriptsize $x_1$}
\psline{|-|}(5.76,1)(6.74,1)\rput(6.25,1.35){\scriptsize $x_2$}
\psline{|-|}(6.76,1)(7.75,1)\rput(7.25,1.35){\scriptsize $x_3$}

\rput(1.6,-1){$L_2^\ell:$}
\psline{|-|}(3.25,-1)(4.24,-1)\rput(3.75,-0.65){\scriptsize $x_1$}
\psline{|-|}(4.26,-1)(5.24,-1)\rput(4.75,-0.65){\scriptsize $x_2$}
\psline{|-|}(5.26,-1)(7.75,-1)\rput(6.50,-0.65){\scriptsize $x_3$}
}

\end{pspicture}

\caption{Example showing how the two lists $L_2^f$ and $L_2^\ell$ store the first and most recent range to cover a position in the text, respectively. The VLG-pattern is $\text{AC} \cdot \gap{1}{5} \cdot \text{T}$. When the occurrence $y$ of $P_2=\text{T}$ at position 9 is reported, we can check the two lists to see that $x_1$ is the first and $x_3$ is the last occurrence of $P_1$ compatible with $y$.\label{fig:twolists}}
\end{figure}
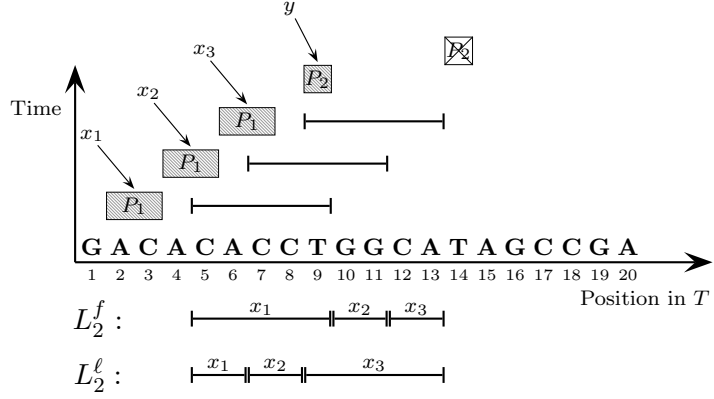

To do this, we replace the list $L_i$ with two lists $L_i^f$ and $L_i^\ell$, for $i=2,\ldots,k$. The idea is that when a position in the text is covered by multiple ranges, $L_i^f$ contains the first range and $L_i^\ell$ contains the most recent range to cover that position. See Fig.~\ref{fig:twolists}. Each range $[s,t]$ in $L_i^f$ or $L_i^\ell$ now also has a reference to the occurrence $x$ of $P_{i-1}$ that defined it, and we will denote the range $[s,t]_x$ to indicate this. When an occurrence $x$ of $P_i$ is reported, we first remove dead ranges from the lists $L_i^f$, $L_i^\ell$, $L_{i+1}^f$ and $L_{i+1}^\ell$ as was done in Algorithm~\ref{alg:ouralgorithm}. If $x$ is relevant a node representing $x$ is added to the implicit gap graph in step~\ref{algstep:appendnode}. In step~\ref{algstep:appendedges}, provided that $x$ is not an occurrence of $P_1$, the two out-going edges of $x$ are added by inspecting $L_i^f$ and $L_i^\ell$ to determine the first and last occurrence of $P_{i-1}$ compatible with $x$. Unless $x$ is an occurrence of $P_k$, the range $R(x)=[\tau + a_i + 1; \tau + b_i + 1]$ is added to the lists $L_{i+1}^f$ and $L_{i+1}^\ell$ in step~\ref{algstep:appendranges} as described in the following section.

\begin{algorithm}
\caption{Algorithm solving the IRVLG problem for a VLG pattern $P$ and a string $T$.\label{alg:irvlgalgorithm}}

\begin{enumerate}
\item Build the AC-automaton for the subpatterns $P_1,P_2,\ldots,P_k$.
\item Process $T$ using the automaton and each time an occurrence $x$ of $P_i$ is reported at position $\tau=\epos(x)$ in $T$ do:
\begin{enumerate}
        \item\label{algstep:removedead} Remove any dead ranges from the lists $L_i^f$, $L_i^\ell$, $L_{i+1}^f$ and $L_{i+1}^\ell$.
        \item\label{algstep:relevancecheck} If $i=1$ or if $\tau - |P_i|=\spos(x)$ is contained in the first range in $L_i^f$ (i.e., $x$ is a relevant occurrence) do:%\hfill {$\triangleright$ \scriptsize \textsf{relevance check}}
        \begin{enumerate}
                \item\label{algstep:appendnode} Add the node $x$ to the implicit gap graph.
                \item\label{algstep:appendedges} If $i > 1$: Add the edges $(x,y)$ and $(x,z)$ to the implicit gap graph, where $y$ and $z$ are the occurrences of $P_{i-1}$ defining the first range in $L_i^f$ and $L_i^\ell$, respectively.
                \item\label{algstep:appendranges} If $i < k$: Let $[q,r]_w$ and $[q^\prime,r^\prime]_{w^\prime}$ denote the first and last range in $L_{i+1}^f$ and $L_{i+1}^\ell$, respectively.
                  \begin{enumerate}
                     \item\label{algstep:maintainlistsA} Append the range $[\max(r+1,\tau + a_i + 1),\tau + b_i + 1]_x$ to the end of $L_{i+1}^f$.
                     \item\label{algstep:maintainlistsB} Change the last range in $L_{i+1}^\ell$ to $[q^\prime, \min(r^\prime, \tau + a_i)]_{w^\prime}$.
                     \item\label{algstep:maintainlistsC} Append the range $[\tau + a_i + 1 ; \tau + b_i + 1]_x$ to the end of $L_{i+1}^\ell$.
                  \end{enumerate}        
        \end{enumerate}
\end{enumerate}
\end{enumerate}
\vspace{-0.3cm}
\end{algorithm}

\subsubsection{Maintaining the Range Lists}
When adding a range $[s,t]_x$ defined by an occurrence $x$ of $P_i$ to $L_{i+1}^f$ and $L_{i+1}^\ell$, we simply append it to the end of the list if it does not overlap the last range in the list. Otherwise, to avoid overlapping ranges, we appropriately shorten either the newly added range $[s,t]_x$ (for $L_{i+1}^f$) or the last range in the list (for $L_{i+1}^\ell$).
The way $L_i^f$ is maintained ensures that the first range that covers some position $\tau$ in $T$ will remain the only range covering this position in $L_i^f$. Conversely, $L_i^\ell$ will store the most recent range covering $\tau$.
In Algorithm~\ref{alg:irvlgalgorithm} the steps \ref{algstep:appendranges}, A, B and C append and possibly shorten the ranges according to this strategy.

\subsubsection{Time and Space Analysis}
As for Algorithm~\ref{alg:ouralgorithm}, the time spent for each of the at most $\alpha$ relevant occurrences reported by the AC automaton is amortized constant. Hence the implicit gap graph can be built in $O(n\log k + m + \alpha)$ time. Storing the implicit gap graph for the entire text takes space $O(\alpha)$, since each of the at most $\alpha$ nodes has at most two out-going edges.

We now consider the space needed to store the lists $L_i^f$ and $L_i^\ell$. The ranges in $L_i^f$ and $L_i^\ell$ are no longer guaranteed to have size $c_{i-1}=b_{i-1}-a_{i-1}+1$ nor being separated by at least one position, so the bound of Lemma~\ref{lem:listsizes} needs to be revised, resulting in a slightly increased space bound for storing the lists. Referring to Fig.~\ref{fig:listsizes}, the number of ranges in $L_i^f$ or $L_i^\ell$ at any point in time is at most
\[
d+1 = c_{i-1}+|P_i|-1+a_{i-1}+1 = |P_i|+b_{i-1}+1 \; .
\]
Summing up, the total space required to store the lists increases from $O(m+A)$ to $O(m+B)$, where $B=\sum_{i=1}^{k-1} b_i$ is the sum of the upper bounds of the lengths of the gaps.

Recapitulating, we have the following theorem
\begin{theorem}\label{thm:extensionIRVLG}
The IRVLG problem can be solved in time $O(n\log k + m + \alpha)$ and space $O(m+B+\alpha)$.
\end{theorem}

\subsection{A Black-Box Solution for Reporting Match Combinations}\label{sec:blackbox}
The number of match combinations, $\beta$, can be exponential in the number of gaps. The implicit gap graph space efficiently encodes all of these match combinations in a graph of size $O(\alpha)$. Thus, a straightforward solution to the RVLG problem is to construct the implicit gap graph and subsequently traverse it to report the match combinations. Each of the $\beta$ match combinations is a sequence of $k$ integers, so this solution to the RVLG takes time $O(n\log k + m + \alpha + k \beta)$ and space $O(m+B+\alpha)$.

We now show that the RVLG problem can be space efficiently solved using any black-box algorithm for the IRVLG problem. The main idea is a simple splitting of $T$ into overlapping smaller substrings of suitable size. We solve the problem for each substring individually and combine the solutions to solve the full problem. By carefully organizing the computation we can efficiently reuse the space needed for the subproblems.

Let $\mathcal{A}_I$ be any algorithm that solves the IRVLG problem in time $t(n,m,k,\alpha)$ and space $s(n,m,k)$, where $n$, $m$, $k$, and $\alpha$, are the parameters of the input as above. We build a new algorithm $\mathcal{A}_R$ from $\mathcal{A}_I$ that solves the RVLG problem as follows. Assume without loss of generality that $n$ is a multiple of $2(m+B)$. Divide $T$ into $z = \frac{n}{m + B} - 1$ substrings $C_1, \ldots, C_z$, called \emph{chunks}. Each chunk has length $2(m+B)$ and overlaps in $m+B$ characters with each neighbor. We run $\mathcal{A}_I$ on each chunk $C_1, \ldots, C_z$ in sequence to compute the implicit gap graph for each chunk. By traversing the implicit gap graph for each chunk we output the union of the corresponding match combinations. Since each match combination of $P$ in $T$ occurs in at most two neighboring chunks it suffices to only store the implicit gap graph for two chunks at any time.

Next we consider the complexity $\mathcal{A}_R$. Let $\alpha_i$ denote the number of occurrences of the strings of $P$ in $C_i$. For each chunk we run $\mathcal{A}_I$ to produce the implicit gap graph. Given these we compute the union of match combinations in $O(k \beta)$ time. Hence, algorithm $\mathcal{A}_R$ uses time 
\[
O\left(\sum_{i = 1}^z t(2(m+B), m, k, \alpha_i) + k \beta\right).
\]
Next consider the space. We only need to store the implicit gap graphs for two chunks at any time. Since the space required for each chunk is $O((m+B)k)$, the total space becomes
\[
O\left((m+B)k + s(2(m+B), m, k)\right).
\]
The black-box algorithm efficiently converts algorithms for the IRVLG problem to the RVLG problem, resulting in the following theorem.
\begin{theorem}
Given an algorithm solving the IRVLG problem in time $t(n,m,k,\alpha)$ and space $s(n,m,k)$, there is an algorithm solving the RVLG problem in time $O\left(\sum_{i = 1}^z t(2(m+B), m, k, \alpha_i) + k \beta\right)$ and space $O\left((m+B)k + s(2(m+B), m, k)\right)$.
\end{theorem}
If we use the result from Theorem~\ref{thm:extensionIRVLG}, we obtain an algorithm that uses time 
\begin{align*}
&O\left(\left(\sum_{i = 1}^z t(2(m+B), m, k, \alpha_i)\right) + k \beta\right) = \\
& \qquad \qquad O\left(\frac{n}{m+B}\left(2(m+B)\log k + m\right ) + \alpha + k \beta \right) = O(n\log k + m + \alpha + k \beta) \; ,
\end{align*}
where the term $m$ in the last expression is needed for the case where $m > n$.
The space usage is
\[
O\left((m+B)k + s(2(m+B), m, k)\right) = O\left((m+B)k + m + B + \max_{i = 1,\ldots, z} \alpha_i\right) = O((m+B)k) \; ,
\]
where the last equality holds, since $a_i \leq (m+B)k$ for all $i$. In summary, we have the following result for the RVLG problem.
\begin{theorem}\label{thm:rvlg}
The RVLG problem can be solved in time $O(n\log k + m + \alpha + k\beta)$ and space $O((m+B)k)$.
\end{theorem}

\subsection{Reporting Match Combinations On the Fly}\label{sec:onthefly}
We now show how a simple extension of our algorithm provides an alternative solution to the RVLG problem achieving the same space and time complexity as the black-box solution. The idea is to use Algorithm~\ref{alg:irvlgalgorithm} and report the match combinations on the fly, while continually removing nodes from the implicit gap graph that no longer can be part of a match combination. We remove the nodes using a method similar to that for removing dead ranges in the lists $L_i^f$ and $L_i^\ell$.

We say that a node $x$ of $P_i$ in the implicit gap graph is \emph{dead} if $x$ can not be part of a future match combination. This happens when
\[
\tau > \epos(x) + \sum_{j=i+1}^k b_{j-1} + |P_j| \; .
\]
Like dead ranges, we can remove dead nodes from the implicit gap graph in amortized constant time. Consequently, all match combinations can be reported in time $O(n\log k + m + \alpha + k\beta)$. Removing the dead nodes ensures that the number of $P_i$ nodes in the implicit gap graph at any time is at most $1+\sum_{j=i+1}^k b_{j-1} + |P_j|$. Thus, the total number of nodes never exceeds
\[
\sum_{i=1}^k 1+\sum_{j=i+1}^k b_{j-1} + |P_j| ~=~ O((m+B)k) \; .
\]
In summary, the algorithm solves the RVLG problem in time $O(n\log k + m + \alpha + k\beta)$ and space $O((m+B)k)$, so it provides and alternative proof of Theorem~\ref{thm:rvlg}.

\end{document}